\documentclass[12pt]{article}

\usepackage{amssymb,amsmath,amsfonts,eurosym,geometry,ulem,graphicx,caption,color,setspace,sectsty,comment,footmisc,caption, pdflscape, subfigure, array, graphicx, subfigure, float, booktabs, appendix, tabularx, mathtools, threeparttable, rotating, lscape, enumerate, bbm, amsthm}

\usepackage[comma]{natbib}

\theoremstyle{definition}
\newtheorem{definition}{Definition}[section]

\newtheorem{proposition}{Proposition}
\newtheorem{lemma}{Lemma}

\renewcommand{\vec}[1]{\mathbf{#1}}
\renewcommand{\vec}[1]{\boldsymbol{#1}}

\newcolumntype{L}[1]{>{\raggedright\let\newline\\arraybackslash\hspace{0pt}}m{#1}}
\newcolumntype{C}[1]{>{\centering\let\newline\\arraybackslash\hspace{0pt}}m{#1}}
\newcolumntype{R}[1]{>{\raggedleft\let\newline\\arraybackslash\hspace{0pt}}m{#1}}

\begin{document}

\begin{titlepage}
\title{Measuring the Input Rank in Global Supply Networks\thanks{We are grateful to participants at the conference on 'Global Value Chains - Recent Advances, Developments, Impacts' in Halle (Germany), the American Economic Association Annual Meeting 2018 in Philadelphia (USA), the Sardinian Empirical Trade Conference 2018 in Cagliari (Italy), the European Trade Study Group 2016 in Helsinki (Finland), the Italian Trade Study Group 2016 in Lucca (Italy), the seminars at the University of Alicante (Spain), and at the Vienna Institute for International Economic Studies (Austria). We want to thank Pol Antràs, Matteo Bugamelli, Italo Colantone, Davide Del Prete, Raja Kali, Kalina Manova, Malwina Mejer, Gianmarco Ottaviano, Frank Pisch, Massimo Riccaboni, Beata Smarzynska, and Lucia Tajoli for valuable comments at different stages of the project. We owe special thanks to Zhen Zhu, who participated to discussions at the onset of the project.}} 
\author{Armando Rungi\thanks{\scriptsize corresponding author. Mail to: armando.rungi@imtlucca.it. Laboratory for the Analysis of Complex Economic Systems, IMT School for Advanced Studies, piazza San Francesco 19 - 55100 Lucca, Italy.} \and Loredana Fattorini\thanks{\scriptsize Mail to: loredana.fattorini@alumni.imtlucca.it. Laboratory for the Analysis of Complex Economic Systems, IMT School for Advanced Studies, piazza San Francesco 19 - 55100 Lucca, Italy.} \and Kenan Huremovic\thanks{\scriptsize Mail to kenan.huremovic@imtlucca.it. Laboratory for the Analysis of Complex Economic Systems,
IMT School for Advanced Studies, piazza San Francesco 19 - 55100 Lucca, Italy.}}
\date{This version: August 2020.} 
\maketitle

\begin{abstract}
\singlespacing
%
\vspace{-15pt}
\noindent We introduce the Input Rank as a measure of relevance of \textit{direct} and \textit{indirect} suppliers in Global Value Chains. We conceive an intermediate input to be more relevant for a downstream buyer if a decrease in that input's productivity affects that buyer more. In particular, in our framework, the relevance of any input depends on: i) the network position of the supplier relative to the buyer, ii) the patterns of intermediate inputs vs. labor intensities connecting the buyer and the supplier, iii) and the competitive pressures along supply chains. After we compute the Input Rank from both U.S. and world Input-Output tables, we provide useful insights on the crucial role of services inputs as well as on the relatively higher relevance of domestic suppliers and suppliers coming from regionally integrated partners. Finally, we test that the Input Rank is a good predictor of vertical integration choices made by 20,489 U.S. parent companies controlling 154,836 subsidiaries worldwide.
\bigskip

\bigskip


\end{abstract}

\setcounter{page}{0}
\thispagestyle{empty}
\end{titlepage}
\pagebreak \newpage

\onehalfspacing

\section{Introduction} \label{sec:introduction}

Modern economies organize as webs of specialized producers within and across national borders. Each company plunges into complex production networks whose configuration is often recursive because the contribution of some intermediate inputs is needed at different stages of completion. Take the case of Audi, a German automobile manufacturer with productive plants around the world. According to the \cite{aip}, it relies upon the deliveries of $1,535$ direct suppliers located in $45$ countries. Among its suppliers, we find Brembo S.p.A., a well-known producer of brakes in Italy, and Bosch Corporation, a subsidiary of the Bosch Group producing valves and other components in Japan. Brembo S.p.A., in turn, reports 33 direct suppliers of parts and components, among which we detect Garrett Motion, a Swiss producer of stamping and aluminum casting. Yet, Garrett Motion is also a direct supplier of Audi, and it delivers components to Bosch Corporation in Japan. After first explorations, we encounter Garrett Motion at least three times in Audi's global supply network, once as a direct and twice as an indirect supplier. 

Similar interlocking patterns are systematic in the global industry, not only in the automotive sector. Consequently, a global fragmentation of production entails a combination of both spider-like and snake-like configurations of supply networks \citep{BaldwinVenables}, both within and across national borders. Yet, the organization of Global Value Chains (GVCs) is mainly modeled assuming a technology of productive tasks over linear sequences, i.e., the 'chain', oriented on upstream-downstream segments \citep{CostinotVogelWang, AntrasChor, FallyHillberry, alfaroetalJPE, antras_degortari}. 

The linear production sequence envisioned in these models allows tractability and enables authors to uncover the primitives of an organization based on mutual contractual interdependence between buyers and suppliers. Final producers meet consumers' demands and realize a surplus distributed with direct and indirect suppliers. Therefore, a blend of contractual environment and market forces shapes the final distribution of the surplus. Inspired by these theoretical models, scholars proposed GVC position metrics, i.e., the Upstreamness and the Downstreamness, which capture a notion of distance of suppliers in the   Input-Output network \citep{alfaroetalJPE, wang, AntrasChor2018, AntrasChor, Fally2012}.

In this paper, we introduce the Input Rank as a bilateral measure of the technological relevance of any, direct or indirect, input in a supply network of a given firm. We start by modeling the problem of a producer, embedded in a network of buyer-supplier relationships, who plans the delivery of her output taking into account the complex web of input-output relations. In this, we build upon recent contributions on production networks \citep{CarvalhoTahbaz-Salehi, GrassiSavaugnat, Baqaee, Grassi2017, Carvalho}. Crucially, in our framework, the most important (direct or indirect) inputs of a firm are the ones that have the potential to affect the marginal costs of that firm the most.

Notably, apart from a full account of the network structure of the supply side of the economy, we allow for a rich heterogeneity concerning how much firms rely on intermediate inputs, as well as competitive forces within each sector. At each stage of production, a lower intensity in the usage of intermediate inputs (i.e., a higher intensity in labor services) buffers the transmission of a shock from upstream markets. Thus, an upstream supplier will rank relatively higher for a given buyer when there are more supply chains (paths) connecting them, and these supply chains are comprised of firms that rely more on intermediate inputs in their production process. At the same time, a higher markup (i.e., a lower competitive pressure) on upstream markets will make downstream buyers more sensitive to input-specific productivity shocks. 



In the second part of our paper, we bring our model to the data. We calibrate it on the U.S. Input-Output tables, sourced from the Bureau of Economic Analysis \citep{usbea}, and on the world Input-Output tables, sourced from WIOD \citep{timmer}. We provide interesting insights into the network dimensions of global production. In particular, we uncover the crucial role of services inputs, which are most central in the configuration of production networks. Then, when we look at geography, we show a pecking order in buyers' sensitivity to input shocks: first domestic suppliers and then suppliers from regionally integrated neighbors (e.g., intra-EU, intra-NAFTA, intra-ASEAN input trade) have the potential to hit harder on downstream buyers.

Finally, we contrast the Input Rank with the Downstreamness \citep{AntrasChor}, Upstreamness \citep{alfaroetalJPE}, and measures based on the Leontief inverse. First of all, we discuss how neither of them considers the endogenous nature of quantities/values reported in I-O tables, where technology blends with market forces to shape the transmission of economic value across industries and countries. Second, we argue that, contrary to the measures mentioned above, the Input Rank accounts for competition structure within sectors. We test the correlation of the Input Rank with choices of vertical integration made by 20,489 U.S. parent companies controlling 154,836 subsidiaries worldwide. We find that a higher Input Rank is positively associated with higher odds that a (direct or indirect) input is vertically integrated. Additionally, we find that parent companies preferably integrate suppliers that are not too distant on the network because they report a relatively lower Upstreamness. 

To grasp the general intuition of inputs' centrality, let us introduce the case of the U.S. economy, which we can plot as a production network\footnote{A bird's eye view of the U.S. production network represented in Figure \ref{fig: IO network} returns an idea of a 'global' outdegree centrality of each industry within a production network. Later, we introduce the Input Rank as a centrality measure of a Katz-Bonacich type (see \cite{bloch2019centrality} for a review of centrality measures in networks). In the U.S. I-O tables, we find a network density of 0.286, i.e., the fraction of actual linkages out of all potential linkages is relatively high. The average path length connecting any two industries is just 1.7 links, pointing to the U.S. economy's small-world nature. Briefly, on average, any producer in an output industry sources inputs from most of the other industries, either directly or indirectly. Indeed, the network of Figure \ref{fig: IO network} is not separable: it is self-contained in a unique connected component where it is always possible to run seamlessly from one node to another by following input linkages.} in Figure \ref{fig: IO network}.

\begin{figure}[H]
\centering
\caption{Input-Output Network from U.S. BEA 2002 I-O tables}
\label{fig: IO network}
\includegraphics[width=0.7\textwidth]{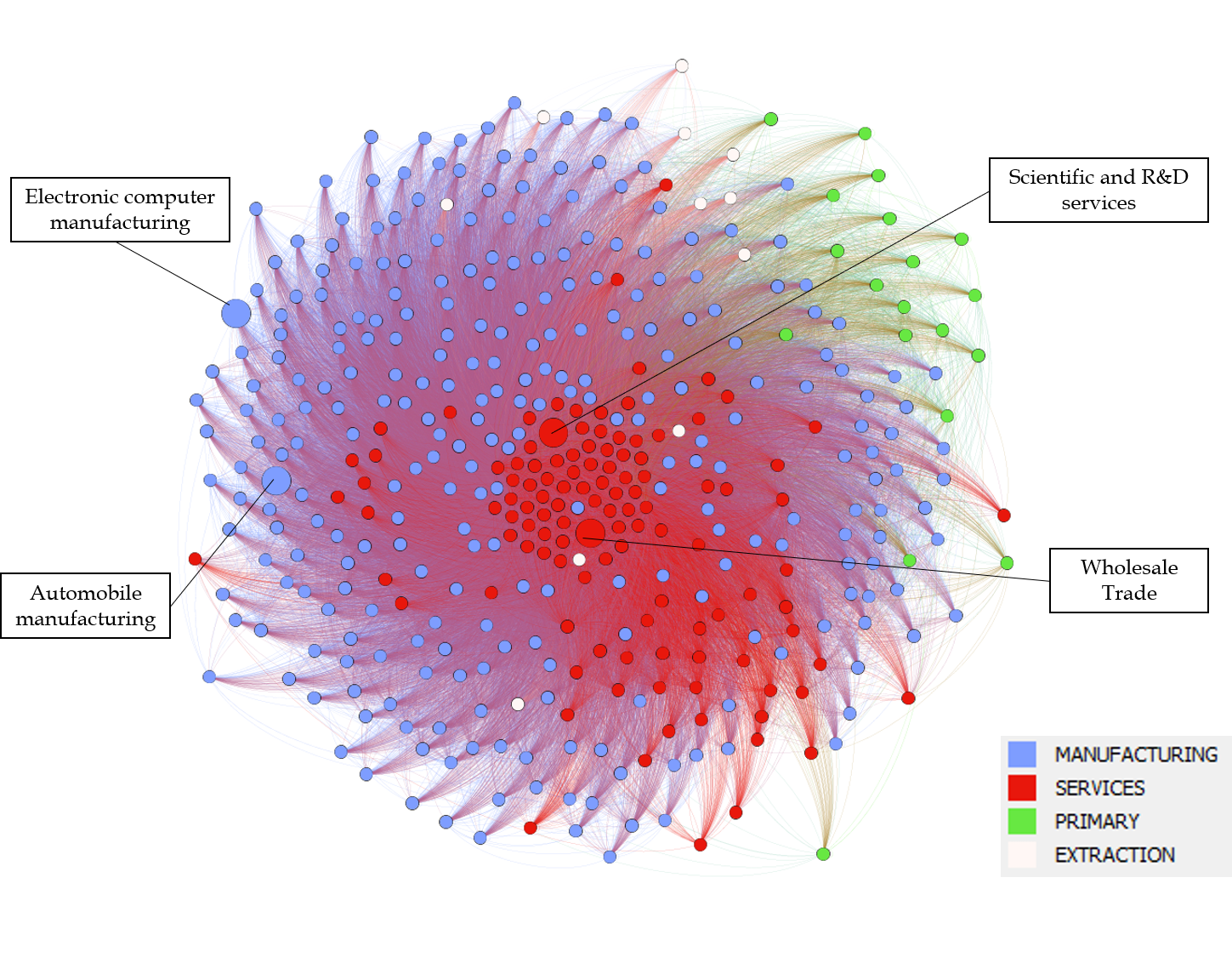}
\begin{tablenotes}
\footnotesize
\singlespacing
\item Note: Nodes represent 425 6-digit NAICS industries from the U.S. BEA 2002 Input-Output tables. Edges represent 51,768 industry-pair transactions. Network density: 0.286. Average path length: 1.7 links. The graph is visualized using a \cite{FruchtermanReingold} layout with the GEPHI software. More connected industries (weighted out-degree) at the center stage. Selected industries in evidence.
\end{tablenotes}
\end{figure}

According to the \cite{usbea} Input-Output tables, we can represent it as a collection of 425 industries (i.e., nodes) linked by 51,768 transactions (i.e., edges). In Figure \ref{fig: IO network}, we organize U.S. industries on a two-dimension space according to their reciprocal connectivity, following a \cite{FruchtermanReingold} layout, which in our case posits more requested inputs at the center stage. Interestingly, services industries make the core of the U.S. production network because they are used as direct inputs in many other (manufacturing and services) industries. On the other hand, primary industries like agriculture and forestry are rather peripheral and mostly located in the north-west area of the graph because their inputs are not as widely used as other inputs. Among services, let us pick the case of R\&D (code 541700) and Wholesale Trade (code 541800), which are among the most connected industries. In fact, wholesalers have a prominent role in professionally distributing many intermediate inputs in different moments of the production process, whereas R\&D services are pivotal in fostering innovation across most U.S. sectors. Now, let us contrast them with two consumer goods industries: Electronic Computer Manufacturing (code 334111) and Automobile Manufacturing (code 336111). They appear to be at the periphery of the U.S. production network because they mostly meet final consumers.

However, once we compare the network positions of selected industries in Figure \ref{fig: IO network} with their positions on the Downstreamness segment \citep{AntrasChor} in Figure \ref{fig: downstreamness IO 2002}, we curiously find that both R\&D and Wholesale Trade are in the middle of an ideally linear supply chain. This is in contrast with the stylized value chain we may have in mind, where a representative business line starts with R\&D services and ends with distribution services\footnote{\cite{alfaroetalJPE} compute a more recent bilateral position metrics, the Upstreamness, which considers the heterogeneity of input positions oriented towards different outputs. However, R\&D services are still on average located in the middle of the output-specific technological sequences, i.e., the average Upstreamness value is 3.044 for an indicator that ranges approximately from 1 to 8.9.}.

\begin{figure}[H]
\centering
\caption{\textit{Downstreamness} from the U.S. BEA 2002 I-O tables}
\label{fig: downstreamness IO 2002}
\includegraphics[width=0.7\textwidth]{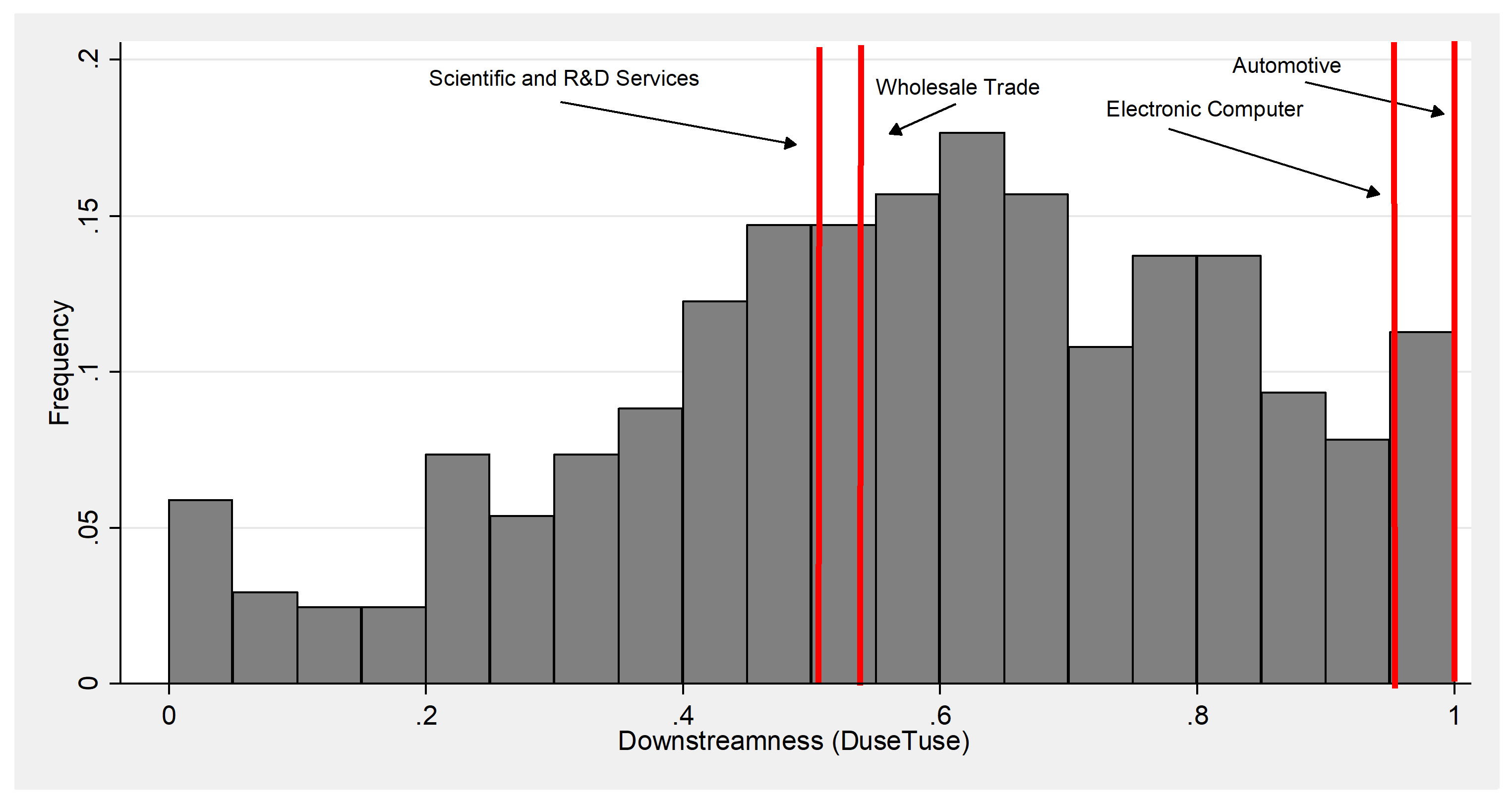}
\begin{tablenotes}
\footnotesize
\singlespacing
\item \textit{Downstreamness (DuseTuse)} is sourced from \cite{AntrasChor}. Frequency indicates how many industries out of a total of 425 from U.S. Input-Output tables in that specific position. Selected industries: Scientific Research and Development Services (code 541700, value 0.504); Wholesale Trade (code 541800, value 0.666); Electronic Computer Manufacturing (code 334111, value 0.959); Automobile Manufacturing (code 336111, value 0.999).
\end{tablenotes}
\end{figure}


The remainder of the paper is organized as follows. The next Section \ref{sec: review} discusses related literature. Section \ref{sec: model} introduces a compact theory for the Input Rank. In Section \ref{sec: applications}, we compute the Input Rank on both the U.S. and world Input-Output tables to describe preliminary evidence. In Section \ref{sec: integration}, we test the role of the Input Rank in firm-level choices of vertical integration. Concluding remarks are offered in Section \ref{sec: conclusions}.

\section{Related literature} \label{sec: review}

A fruitful strand of research emerged in the last decade to study how network dimensions are essential in the organization of production, and how much they contribute to explaining aggregate fluctuations as a response to microeconomic shocks  \citep{GrassiSavaugnat, CarvalhoTahbaz-Salehi, Baqaee, baqaeeetal2018, Oberfield,  acemogluetal2012, Carvalho, acemogluetal2012}. International economics potentially offers a unique environment to study the network dimensions of cross-country trade, investment, and mobility of workers \citep{Chaney2016, Chaney2014}. Yet, the literature on trade and production networks is still in its infancy, and many questions remain unanswered \citep{Bernardetal2018}. 

In this contribution, we focus on bridging with the theory and empirics of Global Value Chains (GVCs), as the latter have been mainly modeled and tested as supposedly linear technological sequences \citep{alfaroetalJPE, FallyHillberry, RungiDelprete, DelpreteRungi, AntrasChor2018, wang, MillerTemurshoev, AntrasChor, Fally2012}, although the existence of spider-like vs. snake-like configurations had been acknowledged by \cite{BaldwinVenables}. The primary purpose of previous literature is to explain how contractual interdependence can shape the organization of GVCs in terms of vertical integration \textit{vis \'a vis} outsourcing strategies. \cite{Berlingierietal2018} recently highlights how the technological importance of inputs on GVCs can be shaped by a solution to both \textit{ex post} contracting problems (transaction costs) and \textit{ex ante} lack of incentives and underinvestment problems (property rights forces).
An implicit step in modeling a network dimension in GVCs has been made by \cite{antras_degortari}, who assume that a linear technology interacts with the geographic centrality of alternative locations. \cite{degortari2019} and \cite{CaliendoParro2015} also exploit implicit information on the network configurations of GVCs to build numerical trade policy counterfactuals based on the transmission of value from inputs to outputs across national borders. 

In our contribution, we introduce the Input Rank as the calibration of a simple network model of production that allows catching the economic interdependence of producers plugged in a  complex supply structure, allowing loops and multiple paths between suppliers and customers. We start from a basic setup whose scope is to study the transmission of microeconomic shocks \citep{CarvalhoTahbaz-Salehi, GrassiSavaugnat, Baqaee, Grassi2017}. Our intuition is that we can frame a problem of input sourcing from a similar perspective. First, a downstream producer observes the topology of her supply structure made of direct and indirect inputs; then, she can assess the impact of a productivity shock of any intermediate input, based on how the shock propagates through the network. 

Input Rank is a recursive measure, and as such is similar to classic eigenvector-type centrality measures like the PageRank \citep{BrinPage}, first used in social networks and search engines, and the \cite{Katz} centrality, first used to assess social status (see \cite{bloch2019centrality} for an extensive discussion and axiomatization of centrality measures used in network analysis). Different variants of eigenvector centrality have been applied to many different domains, from biology and genetics to financial debts, bibliometrics, and road engineering \citep{Gleich, Newman2018}. In either case, a node is more critical in a network if other well-connected nodes point to it, but more distant nodes become less crucial depending on a dampening factor. Similarly, in our production framework, a (direct or indirect) input is more relevant if it delivers (receives) to (from) other highly requested inputs. However, distant suppliers become less and less crucial when labor services are used more intensively along the network, because the downstream buyer relies less on the deliveries by (direct and indirect) suppliers.

\section{A simple model for ranking inputs} \label{sec: model}

In this section, we lay out the theoretical foundations for ranking inputs in supply networks. Assessing the importance of a supplier is not a straightforward task when production processes are fragmented. To illustrate this point, we start with a stylized example depicted in Figure \ref{fig: fictional}, where nodes indicate sectors or representative firms from those sectors (hereafter simply referred to as firms or nodes), whereas directed links indicate deliveries of goods or services.

We focus on the supply chain of firm 1. Consider a scenario in which firm 4 is affected by a distortion or a shock due to which it produces an output of a lower quality/productivity for its customers (firm 3). This affects firm 3's production process since firm 3 uses the output of firm 4 directly in production, and moreover the aforementioned shock is (partially) passed down from firm 3 directly to firm 2, and  both directly and indirectly (through firm 2) to firm 1. Thus, the shock hitting firm 4 is transmitted \textit{downstream} along the network of producers, reaching and affecting firm 1 through different routes. Intuitively, we expect the indirect effect of the shock hitting firm 4 on firm 1 to depend on how much firm 1 relies, directly and indirectly, on the output produced by firm 4. This, in turn, will depend on the relative position of firms 4 and 1 in the network of producers, in particular on the number of weighted paths of \textit{any} length stemming from 4 and reaching firm 1, but also, as will be apparent from the model, the competition structure in sectors along the supply chains.  
\bigskip

\begin{figure}[H]
\centering
\caption{A fictional supply network}
\label{fig: fictional}
\includegraphics[width=0.7\textwidth]{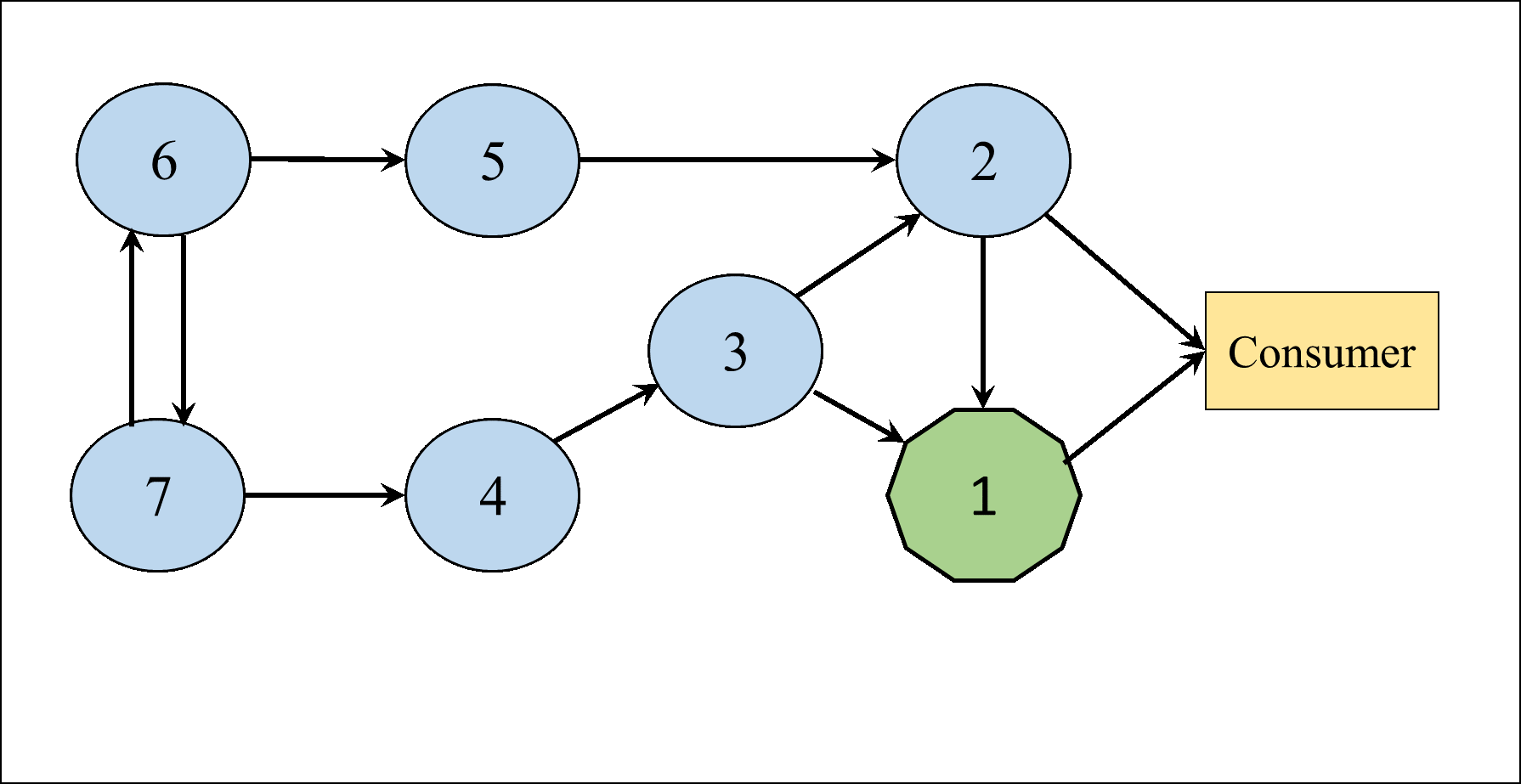}
\end{figure}

In this section, we present a simple general equilibrium model which formalizes the intuition from our simple example. For any given firm operating in sector $k$ in the economy, the model provides a measure of importance of a firm from sector $h$ as a supplier (direct and indirect) of that firm, taking \textit{fully} into account the network structure of the production side of the economy. We call this measure the Input Rank of $h$ relative to $k$. Our theoretical framework is, in many respects, standard in the literature of production networks, and thus we present it in a quite compact manner. In particular we build directly on \cite{Grassi2017} and \cite{Baqaee}. Formal proofs of the claims are relegated to the Appendix.

\subsection{Consumers}

There are two types of agents in the economy: firms and the representative consumer. We denote the set of firms in the economy with $N$. Firms are grouped in $M$ sectors. Each firm belongs to exactly one sector, and it produces a single differentiated variety of a sector-specific good.

The representative consumer owns all ﬁrms in the economy and supplies one unit of labor inelastically. The preferences of the consumer over $M$ goods are defined with the following Cobb-Douglas utility function:
\begin{align}\label{eq:Utility}
U(c_1, c_2, ..., c_M)= C = \theta\prod_{k=1}^{M} c_k^{\gamma_k}
\end{align}
where $c_k$ is the consumption of good $k$, $\sum_{k} \gamma_k =1$, and $\theta = \prod_{k=1}^M \gamma_k^{-\gamma_k}$ is a normalization constant to simplify computations. The composite consumption good $k$ is defined with:
\begin{align}
c_k = \left(\sum_{i=1}^{M_k}c(k,i)^{\frac{\epsilon_k-1}{\epsilon_k}}\right)^{\frac{\epsilon_k}{\epsilon_k-1}},
\end{align}
where $c(k,i)$ denotes the consumption of variety $i$ of good $k$, $\epsilon_k>1$ is the elasticity of substitution across varieties of good $k$, and $M_k$ denotes the number of firms (varieties) in sector $k$. The consumer maximizes her utility subject to the following budget constraint:
\begin{align}
\sum_{k=1}^{M} \sum_{i=1}^{M_k}p(k,i)c(k,i)  = w + \sum_{k=1}^M \sum_{i=1}^{M_k} \pi(k,i),
\end{align}
where $p(k,i)$ is the price of variety $i$ in sector $k$, $\pi(k,i)$ is the profit of firm $i \in k$, and $w$ is the worker's wage bill.

\subsection{Firms}

Firms in sector $k$ are symmetric and use the same constant returns to scale technology that combines labor $\ell$ and intermediate inputs.\footnote{For simplicity labor is the only internal factor of production in the model, i.e., it is not delivered by other firms. Clearly, the model can be extended to include more production factors in a symmetric way without affecting the results.} Each ﬁrm in $k$ produces an imperfectly substitutable variety $i$ of good $k$. Let us denote with $y(k,i)$ the output of firm $i$ from sector $k$, with $\ell(k,i)$ its labor input, and with $x(k,i,h,j)$ the amount of variety $j$ of good $h$ used in the production of variety $i$ of good $k$. Thus, the profit of firm $i$ in sector $k$ is defined with:
\begin{align}
\pi(k,i)&= p(k,i) y(k,i) - \sum_{h=1}^M \sum_{j=1}^{M_h}p(h,j)x(k,i,h,j) - w \ell(k,i)
\end{align}

The production possibilities of a typical firm $i$ from sector $k$ are defined with the following nested production function:
\begin{equation}\label{eq: ProductionFunction}
y(k,i)=\zeta_k\ell(k,i)^{\beta_k} \left(\prod_{h=1}^{M}\left[\sum_{j=1}^{M_h}\left[\tau_h^{-1} x(k,i,h,j)\right]^{\frac{\varepsilon_h-1}{\varepsilon_h}}\right]^{\frac{\varepsilon_h}{\varepsilon_h-1}g_{hk}}\right)^{\alpha_k}.
\end{equation}
where, because of the constant returns to scale assumption, $\alpha_k + \beta_k = 1$ and $\sum_{h}g_{hk} =1$. We normalize $\zeta_k = \beta_k^{-\beta_k}\prod_{h=1}^N g_{hk}^{-\alpha_k g_{hk}}$ to simplify the computations.

Matrix ${\mathbf{G} = (g_{hk})_{h,k=1}^M}$ defines the sector level (technological) production network of the economy, in which the set of nodes is the set of sectors in the economy $\{1,2, ..., M\}$, and two sectors $h$ and $k$ are connected with directed link $h \rightarrow k$ if firms in $k$ use good $h$ as an input in production.  Entry $g_{hk}$ of the adjacency matrix $\mathbf{G}$ represents the weight of link $h \rightarrow k$, where $g_{hk} = 0$ indicates that firms from sector $k$ do not use input $h$ as a (direct) input in production. Given our functional assumptions, $g_{hk} \geq 0$ corresponds to the cost share associated to  intermediate inputs from sector $h$. Therefore, we can directly construct matrix $\mathbf{G}$ using the data from the Input-Output tables. We discuss this in more detail in Section \ref{sec: applications}. 
\bigskip

Parameter $\tau_h \geq 1$ in \eqref{eq: ProductionFunction} captures, in a parsimonious way, the productivity of each variety of good  $h$ when used as an intermediate input. Higher values of $\tau_h$ imply a lower productivity of intermediate input $h$ when used in production. In this paper we are interested in the effects of marginal changes in $\tau_h$ on the equilibrium outcomes. We sometimes refer to these changes as shocks, while we refer to $\tau_h$ as a productivity distortion associated to (inputs produced in) sector $h$.\footnote{For simplicity we assume that distortions and shocks are sector specific. One could of course consider firm-specific distortions $\tau(h,j)$ for $j \in h$, or even pair-specific distortions $\tau(h,j,k)$. We abstract from this type of heterogeneity in the paper.} Intuitively, we think of distortion $\tau_h$ as a failure of firms from sector $h$ to deliver compatible and productive input to their customers on time. Anticipating our later discussion of the vertical integration, we assume that a firm that uses input $h$ either directly or indirectly in its production process can, at least partially, mitigate this distortion by integrating a producer of good $h$. This is one of the reasons we are interested in this particular type of distortions. 
 
It is useful to define the composite intermediate input $h$ used in production of variety $i \in k$ (denoted with $x(k,i,h)$) as an aggregate of varieties of input $h$:
\begin{align}\label{eq:IndividualDemand}
x(k,i,h) = &\left[\sum_{j=1}^{M_h}\left[x(k,i,h,j)\right]^{\frac{\varepsilon_h-1}{\varepsilon_h}}\right]^{\frac{\varepsilon_h}{\varepsilon_h-1}}, 
\end{align}
in which case we can write \eqref{eq: ProductionFunction} as: 
\begin{equation} \label{eq: ProductionFunction1}\tag{\ref{eq: ProductionFunction}A}
y(k,i)=\zeta_k\ell(k,i)^{\beta_k} \left(\prod_{h=1}^{M} (\tau_h^{-1} x(k,i,h))^{g_{hk}}\right)^{\alpha_k }.
\end{equation}

\subsection{Equilibrium}
We assume that firms in the same sector operate in a monopolistic competition environment, and thus set their price to a constant markup over marginal costs. Following \cite{AtkesonBurstein}, we assume that firms set their prices taking as given the other sectors' prices and quantities, the wage, and the aggregate prices and quantities. In Definition \ref{def:Equilibrium} we define the equilibrium concept we are considering.
\bigskip

\begin{definition}[Equilibrium]\label{def:Equilibrium}
A market equilibrium is a collection of prices $p(k,i)$, wage $w$, input demands $x(k,i,h,j)$, outputs $y(k,i)$, consumption $c(k,i)$ and labor demands $\ell(k,i)$ such that:
\vspace{-10pt}
\begin{itemize}
\item[(i)]Each firm $i$ minimizes production costs and applies its mark-up $\mu_i$ to set its price.
\item [(ii)]The representative consumer chooses consumption to maximize utility.
\item [(iii)] Markets for each good and labor clear.
\end{itemize}
\end{definition}
The existence and the uniqueness of the equilibrium follows from standard arguments, see for instance \citep{Grassi2017,Baqaee}.


\subsection{Ranking inputs}\label{subsec:RankingInputs}

Informally, we say that (direct or indirect) input $h$ is more important than input $r$ for firm $i \in k$, or that $i \in k$ relies more on $h$ compared to $r$ if a change $\tau_h$ affects firm $i \in k$ more than the same change in $\tau_r$ affects firm $i \in k$. The first question we ask is if and how the importance of good $h$ for firms in $k$ depends on the underlying structure of sector-level production network. Proposition \ref{prop:PropositionMarginalCostPartial} provides an answer. 
\bigskip


\begin{proposition}\label{prop:PropositionMarginalCostPartial}
Let $\lambda(k,i)$ denote the marginal cost of production of firm $i$ in sector $k$. The following holds in the equilibrium:
\begin{align}\label{eq:MarginalCostCentrality}
\frac{\partial{\log \lambda}(k,i)}{\partial{\log \tau_h}} =  \mathbf{e}_k'\left[\mathbf{I-AG'}\right]^{-1}\mathbf{e}_h = \mathbf{e}_h'\left[\mathbf{I-GA}\right]^{-1}\mathbf{e}_k
\end{align}
where $\mathbf{e}_k$ is the $k$-th unit vector, and $\mathbf{A}$ is the diagonal matrix that collects information about sector specific intermediate inputs' cost shares, $\{\alpha_k\}_{k=1}^{M}$. 

\end{proposition}	
\bigskip

In our simple framework, a negative  shock (distortion) implies an increase in parameter, $\tau_h$. Since the shock cascades through all production paths running towards the downstream buyers, its impact is a function of both the structure of the production network (captured by $\mathbf{G}$), and of the intensities in intermediate inputs, $\{\alpha_k\}_{k=1}^{M}$. Motivated by Proposition \ref{prop:PropositionMarginalCostPartial} we introduce the Input Rank as a measure of the importance of inputs. 
\bigskip



\begin{definition}[Input Rank]\label{def:InputRank} The Input Rank of any upstream supplier of input $h$ relative to the downstream producer of an output $k$ is:
\begin{align}\label{eq: inputrank}
\upsilon_{hk}(\boldsymbol{\mathbf{GA}}) =  \mathbf{e}_h'\left[\mathbf{I-GA}\right]^{-1}\mathbf{e}_k 
\end{align}
\end{definition}
\bigskip

In other words, the bilateral Input Rank, $\upsilon_{hk}$, is the $(h,k)$-th element of the matrix $[\mathbf{I-GA}]^{-1}$. From the perspective of a producer, vector $\boldsymbol\upsilon_k=\{\upsilon_{hk}\}_{h=1}^{M}$ encodes information on the technological relevance of any direct or indirect input $h$ based on its position in the supply structure of producer in the sector $k$. Note that the inverse in \eqref{eq: inputrank} exists since $\mathbf{G}$ is a column-stochastic matrix with spectral radius equal to 1, and $\alpha_k \leq 1$ for each industry $k$.

In Section \ref{sec: integration}, we empirically investigate determinants of a firm's decision to vertically integrate suppliers. Assuming that the integration of an input allows a firm to mitigate\footnote{Please note that our purpose is not to characterize a full incentive structure for the relations between buyers and suppliers. Our model is agnostic with respect to the impact of property rights forces, which have the potential to reduce investment by vertically integrated suppliers, hence reduce productivity.}, at least partly, distortion $\tau_h$ associated to good $h$, we expect that a cost-minimizing firm $i \in k$ is more likely to integrate those inputs $h$ for which a decrease in distortion $\tau_h$ will have the largest positive effect on that firm's marginal cost of production $\lambda(k,i)$. The following proposition is useful as it rationalizes empirical results presented in Table 10.
\begin{proposition}\label{prop:CompStatics}
The marginal effect of a change in $\tau_h$ on the marginal cost of firm $i \in k$ 
is higher for higher values of Input Rank $\upsilon_{hk}$, and for lower values of elasticity of substitution $\epsilon_h$.
\end{proposition}
A higher Input Rank, $\upsilon_{hk}$, implies that suppliers from sector $h$ are relatively more important for firms operating in sector $h$ (in a sense of Proposition \ref{prop:PropositionMarginalCostPartial}), and thus it is not surprising that the effect of a decrease of distortion $\tau_h$ associated to an upstream sector will have a larger effect on the costs of firm $i \in k$ the higher the input rank of that sector relative to $k$. Moreover, the lower degree of substitutability among varieties in sector $h$ amplifies the effect of an input productivity distortion.

In an Appendix Section B, we further introduce a variant of the Input Rank where the downstream buyer imperfectly observes her supply network. We show that, in such a case, there is a potential downplay of more distant nodes, and the downstream buyer overestimates the role of more proximate industries. We believe that cases of imperfect information of the supply network can arise when the contractual environment is imperfect and/or transactions are relation specific, for example in the sense introduced by \cite{rauch1999}, \cite{Nunn2007}, \cite{NunnTrefler2013}, and \cite{NunnTrefler2014}.

\section{Calculating Input Rank from Input-Output tables} \label{sec: applications}

In this Section, we describe how Input Rank can be calculated using the Input-Output (I-O) tables. We contrast the Input Rank with well-known metrics of sequential production, i.e., the Downstreamness and the Upstreamness. Finally, we provide some insights on inputs' centrality in the organization of global production. For consistency with previous studies, we compute the Input Rank on both the U.S. and world Input-Output tables.

To calculate Input Rank using I-O, we first derive the demand of intermediate inputs at the sector level. Aggregating individual demands $x(h,i,k)$ from Eq. \eqref{eq:IndividualDemand}, we show that (see Lemma \ref{lem:AggregationDemand} in Appendix A)
\begin{align*}
x(k,h) = \frac{\alpha_k g_{hk}p_ky_k}{\mu_k p_h} =  \frac{\epsilon_k-1}{\epsilon_k}\frac{\alpha_k g_{hk}p_ky_k}{p_h}.
\end{align*}
Let us define $Z_{hk} \equiv p_{h}x(k,h)$, which is the value of input deliveries from sector $h$ to sector $k$ and is readily available in the I-O data. From \eqref{eq:SectorLevelDemand} it directly follows that
\begin{align}
\frac{p_h x(k,h)}{\sum_{\ell}p_{\ell} x(k,\ell)} = \frac{Z_{hk}}{\sum_{\ell}Z_{\ell k}}  = g_{hk},
\end{align}
which describes a way to construct $\mathbf{G}$ from I-O data. 
Thanks to the Cobb-Douglas assumption, elements $\{\alpha_k\}_{k=1}^{M}$ of the diagonal matrix $\mathbf{A}$ can be calculated as cost shares of intermediate inputs at the sector level. To be more precise, we calculate $\alpha_k$ in the following way:
\begin{align*}
\alpha_k = \frac{\text{Cost of intermediate inputs of sector $k$}}{\text{Total cost on all inputs of sector $k$}}.
\end{align*}
 
The Input Ranks are then calculated by inverting matrix $\mathbf{I - GA}$.

\subsection{Relation with other I-O based measures}

The Input-Rank, as defined in \eqref{eq: inputrank}, bears some similarity to other commonly used measures based on I-O network including the Upstreamness, the Downstreamness and the total input requirements. In this section we compare the Input Rank to these measures, and argue that it captures additional information ignored by these measures (i.e. within-sector competition structure).

Before doing so, let us introduce some additional notation. Let $Y_k \equiv p_k y_k $ denote the value of gross output of sector $k$, and let $d_{hk}$ denote the direct requirement coefficient obtained from I-O tables and equal to the ratio of value  of sales from sector $h$ to sector $k$  ($Z_{hk}$), and the value of the gross output of sector $k$ ($Y_k$). The matrix $(\mathbf{I-D})^{-1}$ is known as the \textit{Leontief inverse} or \textit{the total requirements matrix}, where $\mathbf{D}$ is the matrix of direct requirements. \cite{AntrasChor2018} discusses a number of network statistics based on the Leontief inverse and its square $\mathbf{(I-D)^{-2}}$, including the Downstreamness and the Upstreamness measures, introduced in \cite{Fally2012, AntrasChor, MillerTemurshoev}.

There are at least two important differences between the Input Rank and these measures. First, even though the Input Rank matrix $(\mathbf{I - G A})^{-1}$ is reminiscent of the Leontief inverse $(\mathbf{I-D})^{-1}$,  matrices $\mathbf{ G A}$ and $\mathbf{D}$ are in general different. To see this, note that from \eqref{eq:SectorLevelDemand} it directly follows that:
\begin{align}
d_{hk} = \frac{Z_{hk}}{Y_k}=\frac{p_{h}x(k,h)}{p_k y_k}= \frac{\epsilon_k-1}{\epsilon_k}\alpha_k g_{hk},
\end{align} 
and hence $\mathbf{GA} = \mathbf{DM}$, where $\mathbf{M}$ is a diagonal matrix of markups, $\mu_k = \frac{\epsilon_k}{\epsilon_{k}-1}$. Thus, contrary to the Leontief inverse, the Input Rank explicitly, through markups, takes into account the within-sector market structure. The two measures coincide only in the special case when there is a perfect competition in \textit{each} sector.\footnote{Although we assume monopolistic competition within each sector, our model can be easily extended to different types of within-sector imperfect competition, in which case the exact expression for markups changes. See \cite{Baqaee} for more details.} Since elements of the diagonal matrix $\mathbf{M}$ are greater than one, the Input Rank implies \textit{more} interdependence between sectors than the Leontief inverse. Second, the Leontief inverse as well as the above mentioned measures of Downstreamness and Upstreamness are derived from accounting identities implied by the I-O tables. The Input Rank measure is a direct result of a (standard) general equilibrium model that explicitly takes into account the network nature of the production side of the economy.

To illustrate the empirical difference between the Input Rank and the related measures discussed above, we report rank correlations in Table \ref{tab: Spearman}. Usefully, both the Spearman's and the \cite{Kendall}'s correlations assess how well the relationship among I-O metrics can be described by a monotonic function, i.e., when inputs have a similar ordinal position across metrics. Interestingly, we observe that there is a low and weakly significant rank correlation between the Input Rank and the Upstreamness, albeit with a negative sign given the upstream orientation of the latter (-.40 and -.27, respectively). In fact, the rank correlations of Input Rank with Upstreamness are lower than the ones with the direct coefficients, at the bottom of the table (.55 and .30, respectively). In this respect, the Upstreamness has a relatively higher association with original information from I-O tables (-.68 and -.39, respectively) than with the Input Rank. 
\begin{table}[H]
\centering
\doublespacing
\caption{Rank correlations}\label{tab: Spearman}
\resizebox{0.6\textwidth}{!}{%
\begin{tabular}{lcc}
\toprule
&\hspace{6pt} Input Rank ($\upsilon_{hk}$) &\hspace{6pt} Upstreamness ($upstr_{hk}$)  \\
Upstreamness ($upstr_{hk}$) &  & \\
\hspace{6pt} Spearman's $\rho$ & -.3962* & \\
\hspace{6pt} Kendall's $\tau$ & -.2672* & \\
 I-O coefficient ($d_{hk}$) & & \\
\hspace{6pt} Spearman's $\rho$ & .5468*& -.6789*\\
\hspace{6pt} Kendall's $\tau$ & .3040* & -.3875*\\
\addlinespace
\hline
\bottomrule
\end{tabular}%
}\\
\begin{tablenotes}
\footnotesize
\singlespacing
\item Note: Spearman's and Kendall's rank correlations are computed between the Input Rank, the Upstreamness following \cite{alfaroetalJPE}, and the direct requirement coefficients, as originally sourced from U.S. 2002 I-O tables. All measures are bilateral for $88,595$ input-output pairs. * stands for p-value $< .1$.
\end{tablenotes}
\end{table}
Eventually, given also evidence reported in Table \ref{tab: Spearman}, we conclude that the Input Rank, the Upstreamness, and the direct requirement coefficients from I-O tables convey different information on I-O linkages.

\subsection{U.S. Input Output tables}

U.S. I-O 2002 tables, compiled by the Bureau of Economic Analysis (BEA), sketch a reasonably fine-grained supply network established among 6-digit industries. The same tables have been extensively used to study production networks \citep{Carvalho}, vertical integration choices \citep{Acemogluetal2009, Alfaroetal2016}, and to compute Upstreamness/Downstreamness metrics \citep{alfaroetalJPE, AntrasChor}. In Figure \ref{fig: IO network}, we already showed how a solid and complex production network emerges from a visualization of U.S. I-O tables. After a closer look, we register a strong heterogeneity in sourcing strategies at the industry level. For example, in Appendix Figures \ref{fig: indegree US} and \ref{fig: outdegree US}, we report both the in-degree and out-degree distributions by industry, i.e., the number of inputs received and the deliveries made by each node of the U.S. production network. On average, the in-degree of an industry is higher than its out-degree. As expected, the industry with the highest number of input industries (296) is the Retail Trade (code 4A0000), because retailers professionally sell physical goods to consumers. On the other hand, the industry with the highest number of purchasing industries (425) is the Wholesale Trade (code 420000), because wholesalers professionally distribute intermediate physical inputs to all industries. Yet, ‘global’ centralities measured by in- or out-degrees are of scarce interest to understand the ‘local’ role of an upstream industry with respect to each downstream output. 

More properly, the Input Rank shall return the technological relevance of an input market for a downstream producer active in an output market. In Figure \ref{fig: matrix}, we visualize the results from the computation of the Input Rank as a matrix of industry-pair values. A darker cell implies that an input industry is more technologically relevant for the producers in a specific output industry. Interestingly, from the upper part of the matrix, we find that services industries are the most relevant across many manufacturing and services industries. Please note that the diagonal is always dark, since most producers source an important amount of inputs within their industry, hence intra-industry suppliers are also technologically relevant.

\begin{figure}[H]
\centering
\caption{Input Rank computed on U.S. 2002 Input-Output tables}
\label{fig: matrix}
\includegraphics[width=0.99\textwidth]{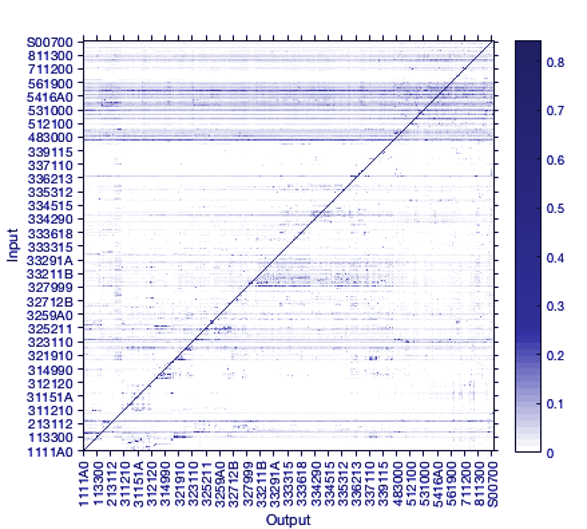}
\begin{tablenotes}
\footnotesize
\singlespacing
\item Note: Input Rank vectors are computed for each using industry among 421 industries classified at the 6-digit in the U.S. BEA 2002 tables. A darker cell implies that an input is more technologically relevant for an output.
\end{tablenotes}
\end{figure}

\begin{table}[H]
\centering
\caption{Top 20 inputs}
\label{tab:top20}
\resizebox{.85\textwidth}{!}{%
\begin{tabular}{cllcc}
\toprule
\addlinespace
Abs. Rank & I-O code & Industry name & avg. Input Rank & st. dev.\\
\hline
1 & 531000 & Real Estate & .03442 & .03390\\
2 & 550000 & Management of Companies and Enterprises & .02369 & .02977\\
3 & 420000 & Wholesale Trade & .02067 & .02903\\
4 & 211000 & Oil and Gas Extraction & .01799 & .01378\\
5 & 324110 & Petroleum Refineries & .01425 & .01342\\
6 & 221100 & Electric Power Generation, Transmission, and Distribution & .01416 & .02026\\
7 & 533000 & Lessors of Nonfinancial Intangible Assets & .01383 & .00770\\
8 & 517000 & Telecommunications & .01350 & .02307\\
9 & 52A000 & Monetary Authorities and Depository Credit Intermediation & .01263 & .02721\\
10 & 541800 & Advertising and Related Services & .01254 & .02084\\
11 & 331110 & Iron and Steel Mills and Ferroalloy Manufacturing & .01216 & .02891\\
12 & 484000 & Truck Transportation & .01017 & .01935\\
13 & 325190 & Other Basic Organic Chemical Manufacturing & .00922 & .01273\\
14 & 523000 & Securities, Commodity Contracts, Investments, and Related Activities & .00870 & .02910\\
15 & 221200 & Natural Gas Distribution & .00783 & .00779\\
16 & 325211 & Plastic Material and Resin Manufacturing & .00736 & .01384\\
17 & 522A00 & Nondepository Credit Intermediation and Related Activities & .00729 & .02284\\
18 & 722000 & Food Services and Drinking Places & .00707 & .02276\\
19 & 230301 & Nonresidential Maintenance and Repair & .00682 & .02646\\
20 & 541610 & Management, Scientific, and Technical Consulting Services & .00645 & .02803\\
\bottomrule
\end{tabular}}
\begin{tablenotes}
\footnotesize
\singlespacing
\item Note: Average values computed for usage across using industries. Input Rank vectors are computed for each using industry among 421 industries classified at the 6-digit in the U.S. BEA 2002 tables. \end{tablenotes}
\end{table}

\begin{table}[H]
\centering
\caption{Top 20 manufacturing inputs }
\label{tab:top20manuf}
\resizebox{.85\textwidth}{!}{%
\begin{tabular}{cllcc}
\toprule
\addlinespace
Abs. Rank & I-O code & Industry name & avg. Input Rank & st. dev.\\
\hline
5 & 324110 & Petroleum Refineries & .01425 & .01342\\
11 & 331110 & Iron and Steel Mills and Ferroalloy Manufacturing & .01216 & .02084\\
13 & 325190 & Other Basic Organic Chemical Manufacturing & .00922 & .01273\\
16 & 325211 & Plastic Material and Resin Manufacturing & .00736 & .01384\\
21 & 336300 & Motor Vehicle Parts Manufacturing & .00645 & .01634\\
26 & 334413 & Semiconductor and Related Device Manufacturing & .00560 & .01530\\
27 & 334418 & Printed Circuit Assembly Manufacturing & .00525 & .01347\\
28 & 325110 & Petrochemical Manufacturing & .00502 & .00708\\
29 & 32619A & Other Plastics Product Manufacturing & .00497 & .01531\\
33 & 322210 & Paperboard Container Manufacturing & .00466 & .01181\\
38 & 321100 & Sawmills and Wood Preservation & .00448 & .01495\\
41 & 323110 & Printing & .00410 & .01906\\
43 & 322120 & Paper Mills & .00403 & .01360\\
48 & 3259A0 & All Other Chemical Product and Preparation Manufacturing & .00365 & .01077\\
50 & 326110 & Plastics Packaging Materials and Unlaminated Film and Sheet Manufacturing & .00349 & 01280\\
52 & 332710 & Machine Shops & .00345 & .02262\\
54 & 33131A & Alumina Refining and Primary Aluminium Production & .00331 & .014346\\
55 & 322130 & Paperboard Mills & .00315 & .01106\\
63 & 332800 & Coating, Engraving, Heat Treating and Allied Activities & .00290 & .01739\\
64 & 331411 & Primary Smelting and Refining of Copper & .00283 & .01312\\
\bottomrule
\end{tabular}}
\begin{tablenotes}
\footnotesize
\singlespacing
\item Note: Average values computed for usage across using industries. Input Rank vectors are computed for each using industry among 421 industries classified at the 6-digit in the U.S. BEA 2002 tables. Absolute rank indicates the position when considering also services inputs.
\end{tablenotes}
\end{table}

In Table \ref{tab:top20}, we report the average Input Rank for top 20 inputs across all I-O industries with standard deviations. Here, as well, we find that services industries are, on average, ranked higher than manufacturing. The top three inputs are Real Estate (code 531000), Management of Companies and Enterprises\footnote{This sector mainly gathers headquarters services by holding firms. As from the original definition \citep{BLS2018}: 'This sector comprises: i) companies that hold financial activities (securities or other equity interests) in other companies for the purpose of a corporate control to influence management decisions; ii) companies that professionally administer, oversee, and manage other companies through strategic or organizational planning and decision making.'} (code 550000), and Wholesale Trade (code 420000). 

Further down, we find energy including fossil fuels (codes 211000, 324110, and 221200) and electricity (code 221100). Royalties (code 533000), telecommunications (code 517000), the financial sector (codes 52A000, 523000, and 522A00), and R\&D services (code 541610) are also relatively more important than many manufacturing industries. In fact, when we look at top manufacturing inputs in Table \ref{tab:top20manuf}, we have to go deep further down the rank with lower magnitudes and rising standard deviations. The ratio between the Input Ranks of top first and top 20th input in both Tables \ref{tab:top20} and \ref{tab:top20manuf} is above 5, indicating that there is a relatively high skewness of the Input Rank distributions.

Clearly, values of the Input Rank can be much heterogeneous across downstream producers. In Appendix Tables \ref{tab: top10 auto} and \ref{tab: top10 computer}, we look from the perspective of selected root industries (Electronic Computer Manufacturing, code 334111; Automobile Manufacturing, code 336111), and we find that top ten inputs by technological centrality always include a combination of manufacturing inputs that are predictably specific for their production processes, and a selection of services that rank higher despite their lower I-O technical requirements 

\subsection{International Input-Output tables}

In this Section, we report computations of the Input Rank on world Input-Output Tables by WIOD. World I-O tables have been extensively used in settings where the geographical dimension of GVCs is important (among others, see \cite{worldbank, FajgelbaumKhandelwal2016, Adaoetal2017, CaliendoParro2015, CostinotRodriguezclare2014}, where the main issue has been to track value added generation by country and industry while avoiding a classical problem of multiple accounting from standard international trade data. For a useful review, see \cite{Johnson2018}.

WIOD data report information on exchanges among 43 countries and 56 ISIC rev. 4 2-digit industries in the period 2000 - 2014. A model for the rest of the world balances world trade. Additional socio-economic accounts contain the information we need to calculate labor and intermediate shares. For further details, see \citep{timmer}. 

After computing the Input Rank on WIOD in year 2014, we report averages and standard deviations for top 10 supplying countries in Table \ref{tab: top10 world partners}. On the podium, we find China, Germany, and the USA. Notably, Russia is fourth thanks to the shipment of primary inputs and fuels, which are technologically relevant inputs across many countries and industries. Italy ranks also high thanks to its centrality within many manufacturing supply chains, as producers both import and export intermediate inputs with many partners.

When we report averages and standard deviations for top input industries in Table \ref{tab: top10 world industries}, we find on the podium: i) Mining and Quarrying (code B); ii) Wholesale Trade (code G46); iii) Electricity, Gas, Steam and Air Conditioning Supply (code (D35). Similarly to what observed in the case of U.S. I-O tables, we find that manufacturing inputs have a relatively lower technological relevance than services and primary industries.


\bigskip
\begin{table}[H]
\centering
\caption{Top 10 supplying countries by Input Rank in the World I-O tables (WIOD).}
\label{tab: top10 world partners}
\resizebox{.6\textwidth}{!}{%
\begin{tabular}{clcc}
\toprule
\addlinespace
Abs. Rank & Origin country & Avg. Input Rank & st. dev.\\
\hline
1 & China & .00352 & .02773\\
2 & Germany & .00310 & .02409\\
3 & USA & .00275 & .02403\\
4 & Russia & .00218 & .02389\\
5 & Italy & .00206 & .02403\\
6 & UK & .00178 & .02395\\
7 & France & .00168 & .02390\\
8 & S. Korea & .00150 & .02489\\
9 & Poland & .00146 & .02386\\
10 & Turkey & .00140 & .02545\\
\bottomrule
\end{tabular}}
\begin{tablenotes}
\footnotesize
\singlespacing
\item Note: The full Input Rank matrix on WIOD data has been computed starting from 43 countries \textit{per} 56 industries possible origins of sourcing. Average values and standard deviations reported for every triplet including: destination country, destination industry, sourcing industry.
\end{tablenotes}
\end{table}

\bigskip
\begin{table}[H]
\centering
\doublespacing
\caption{Top 10 input industries by Input Rank in the World I-O tables (WIOD).}
\label{tab: top10 world industries}
\resizebox{0.99\textwidth}{!}{%
\begin{tabular}{cllcc}
\toprule
\addlinespace
Abs. Rank & I-O code & Origin industry & Avg. Input Rank & st. dev.\\
\hline
1 & B & Mining and Quarrying & .00343 & .02810\\
2 & G46 & Wholesale Trade & .00331 & .02545\\
3 & D35 & Electricity, Gas, Steam and Air Conditioning Supply & .00309 & .03349\\
4 & C20 & Manufacture of Chemicals and Chemical Products & .00286 & .02558\\
5 & N & Administrative and Support Service Activities & .00285 & .02620\\
6 & K64 & Financial Service Activities & .00275 & .02760\\
7 & C19 & Manufacture of Coke and Refined Petroleum Products & .00252 & .02403\\
8 & M69\&M70 & Legal and Accounting Activities; Head Offices; Management Consultancy & .00249 & .02508\\
9 & C24 & Manufacture of Basic Metals & .00242 & .02598\\
10 & H49 & Land Transport and Transport via Pipelines & .00227 & .02501\\
\bottomrule
\end{tabular}}
\begin{tablenotes}
\footnotesize
\singlespacing
\item Note: The full Input Rank matrix on WIOD data has been computed starting from 43 countries \textit{per} 56 industries possible origins of sourcing. Average values and standard deviations reported for every triplet including: destination country, destination industry, sourcing country.
\end{tablenotes}
\end{table}
\bigskip

However, please note that ISIC 2-digit industries in WIOD aggregate much more than NAICS 6-digit industries in U.S. I-O tables, thus confounding several intermediate inputs, on one hand, and intermediate inputs with final goods in the same category, on the other hand. For this reason, we prefer keeping U.S. tables for following exercises on vertical integration (Section \ref{sec: integration}), where our interest properly falls on the technological relevance of an upstream industry wherever the origin country for sourcing. 

Nonetheless, we can still comment on applications to the WIOD data to describe the relevance of trading partners in general, with a geographic disaggregation, as providers of inputs to producers in other countries. For this reason, in Table \ref{tab: top10 3 countries}, we look at the sourcing strategies of three countries, USA, China, and Germany, which were also classified as most relevant suppliers in Table \ref{tab: top10 world partners}. Brefly, this time we are interested in ranking their providers of inputs. 

Interestingly, on top of the rank, in each case, we find the country itself. Domestic producers are by far the most relevant, according to our Input Rank, as a domestic shock have the the potential to hit relatively more than foreign shocks. In second places, China is a crucial partner for both USA and Germany, whereas South Korea is the most central foreign partner for China. Further down in the rankings, we notice a predominance of regional trading partners, possibly located in the same region/continent. 

\newpage

\begin{landscape}
\vspace*{45pt}
\begin{table}[H]
\centering
\doublespacing
\caption{Top 10 origin countries by Input Rank for USA, China, and Germany}
\label{tab: top10 3 countries}
\resizebox{1.35\textwidth}{!}{%
\begin{tabular}{cllcllcclcc}
\toprule
\addlinespace
& & \large USA & & &\hspace{15pt} \large China & & &\hspace{15pt} \large Germany & & \\
\addlinespace
\hline
Abs. Rank &\hspace{15pt} &  Origin country & Avg. Input Rank & st. dev &\hspace{15pt} Origin country & Avg. Input Rank & st. dev &\hspace{15pt} Origin country & Avg. Input Rank & st. dev\\ \\
1 &\hspace{15pt} & USA & .04166 & .14990 &\hspace{15pt} China & .05793 & .17196 &\hspace{15pt} Germany & .03822 & .15154\\
2 &\hspace{15pt} & China & .00167 & .00327 &\hspace{15pt} S. Korea & .00068 & .00211 &\hspace{15pt} China & .00185 & .00344\\
3 &\hspace{15pt} & Canada & .00123 & .00419 &\hspace{15pt} USA & .00062 & .00096 &\hspace{15pt} Netherlands & .00162 & .00507\\
4 &\hspace{15pt} & Mexico & .00073 & .00191 &\hspace{15pt} Japan & .00053 & .00130 &\hspace{15pt} USA & .00149 & .00237\\
5 &\hspace{15pt} & Germany & .00048 & .00089 &\hspace{15pt} Australia & .00051 & .00166 &\hspace{15pt} Russia & .00101 & .00319\\
6 &\hspace{15pt} & Japan & .00046 & .00127 &\hspace{15pt} Taiwan & .00040 & .00203 &\hspace{15pt} France & .0010 & .00203\\
7 &\hspace{15pt} & Korea & .00037 & .00090 &\hspace{15pt} Russia & .00039 & .00105 &\hspace{15pt} Italy & .00094 & .00200\\
8 &\hspace{15pt} &  UK & .00034 & .00049 &\hspace{15pt} Germany & .00037 & .00064 &\hspace{15pt} UK & .00081 & .00137\\
9 &\hspace{15pt} &  Russia & .00026 & .00067 &\hspace{15pt} Brazil & .00026 & .00063 &\hspace{15pt} Poland & .00066 & .00180\\
10 &\hspace{15pt} &  France & .00022 & .00041 &\hspace{15pt} Indonesia & .00015 & .00041 &\hspace{15pt} Belgium & .00065 & .00180\\
\hline
\bottomrule
\end{tabular}}
\begin{tablenotes}
\footnotesize
\singlespacing
\item Note: Input Rank vectors for USA, China and Germany have been computed starting from 43 countries \textit{per} 56 industries possible origins of sourcing. Average values and standard deviations reported for within-country industry-level Input Ranks.
\end{tablenotes}
\end{table}
\end{landscape}
\newpage

\subsection{A pecking order in the geography of sourcing}

To check whether the country studies showed in Table \ref{tab: top10 3 countries} present regularities that can be extended to other countries, we test dummy regressions in Table \ref{tab: regional integration} and repeat the exercise first for all industries and then for manufacturing industries only. 

We consider the (log of the) Input Rank as dependent variables and two binary variables as regressors. The idea is to catch the premia in technological centrality by the geography of sourcing. The first indicator is equal to one if the (direct or indirect) supplier are $Domestic$, and zero otherwise. A second binary indicator is equal to one if the (direct or indirect) suppliers come from the same region/continent of the final buyers, i.e., they are $Intraregion$. Two-way fixed effects for supplying and buying industries are included. Standard errors are clustered multiway following \cite{Cameronetal2011} and including the quartet of countries and industries in origin and destination. The exercise is repeated for EU members, former NAFTA members, and for Asian countries\footnote{Please note that we included Australia in the group of Asian countries present in WIOD data (China, Indonesia, India, Japan, South Korea, Taiwan) because the country is among the biggest trading partners of all of them. Moreover, all these countries, including Australia, participate in the ASEAN Plus Six, an enlarged version of the Association of the South-East Asian Nations, which promotes regional economic integration.}

Interestingly, the results clearly indicate a pecking order in the geography of upstream markets. Domestic producers rank relatively higher than intra-regional ones, which in turn rank higher than all other suppliers. In fact, we find that domestic (direct or indirect) suppliers are on average 5 times, 3.2 times, and 4.7 times more central for a representative buyer in the European Union, the former NAFTA\footnote{The North-American Free Trade Agreement (NAFTA) has been substituted by the United States–Mexico–Canada Agreement (USMCA). After renegotiations in the period 2017/2018, the USMCA entered into force in all member states on July 1, 2020.}, and Asian countries, respectively. Despite high trade openness, domestic producers are still by far the most technologically relevant for most producers. 

On the other hand, if any shock occurs on foreign upstream markets, our results show that final producers will be hit harder if they operate under the same regional trade agreements. The premium on centrality is, on average, 1.6 and 2 times more in the case of intra-EU and intra-NAFTA shipments. It is relatively less ($0.85\%$) but still positive and statistically significant in the case of Asian producers.

\begin{table}[H]
\centering
\singlespacing
\caption{Input Rank: a pecking order on geography}
\label{tab: regional integration}
\resizebox{.99\textwidth}{!}{%
\begin{tabular}{lcccccc}
\toprule
\addlinespace
Dependent variable & EU & EU & NAFTA & NAFTA & ASEAN & ASEAN \\
(log of) Input Rank & & & & & &\\
\addlinespace
\hline
$Domestic$ & 5.053*** & 4.299*** & 3.196*** & 2.078* & 4.722***& 4.042*** \\
 & (.281) & (.302) & (.703) & (.569) & (.383) & (.430) \\
$Intraregion$ & 1.596*** & 1.642*** & 1.984*** & 2.488*** & .849***& .739**\\
 & (.193) & (.171) & (.151) & (.075) & (.252) & (.221) \\
$Constant$ & -13.226*** & -12.758*** & -11.251*** & -10.630*** & -11.550*** & -10.585***\\
 & (.373) & (.402) & (1.105) & (.806) & (.373) & (.478)\\
\hline
N. obs. & 3,486,124 & 394,856 & 360,797 & 42,390 & 965,888 & 98,910\\
Activities & All & Manuf Only & All & Manuf Only & All & Manuf Only\\
N. origin countries & 43+1 & 43+1 & 43+1 & 43+1 & 43+1 & 43+1 \\
N. destination countries & 28 & 28 & 3 & 3 & 7 & 7\\
Origin industry fe & Yes & Yes & Yes & Yes & Yes & Yes\\
Destination industry fe & Yes & Yes & Yes & Yes & Yes & Yes\\
Multiway clustered errors & Yes & Yes & Yes & Yes & Yes & Yes\\
Adj. R squared & .366 & .244 & .344 & .363 & .489 & .393\\
\hline
\bottomrule
\end{tabular}%
}
\begin{tablenotes}
\footnotesize
\singlespacing
\item Note: Cross-section data for year 2014 are sourced from WIOD data \citep{timmer}, including 43 countries, 56 industries, and a closure for the 'rest of the world'. Fixed effects by supplying and buying industries. Errors are clustered multiway following \cite{Cameronetal2011}, considering origin country, supplying industry, destination country, and buying industry.  ***, **, * stand for $p < .01$ and $p < .05$, and $p < .1$, respectively.  
\end{tablenotes}
\end{table}

When we restrict our analysis to manufacturing industries only, in columns 2, 4, and 6 of Table \ref{tab: regional integration}, we find that there is a lower albeit still high centrality of domestic suppliers in all cases. Yet, in the case of former NAFTA, intraregional manufacturing suppliers are, on average, about 2.5 times more central than other foreign suppliers. The latter finding is driven by a high degree of regional integration along the manufacturing supply chain developed across the U.S., Canadian, and Mexican borders (e.g., in the automotive industry), making all producers most sensitive to upstream shocks. In fact, the latter is the highest premium on technological centrality that we can recover in Table \ref{tab: regional integration}.

\section{The role of the Input Rank in vertical integration} \label{sec: integration}

The decision to \textit{make or buy} an input is the typical situation when a producer needs gathering information on the technological relevance of both direct and indirect inputs. In this Section, we test whether the Input Rank is a good predictor of vertical integration choices. The intuition is that vertical integration possibly allows mitigating the frictions from upstream markets. More in general, we expect that a parent company can internalize part of the shocks coming from (direct or indirect) suppliers. In this case, technologically central inputs in the sense we discussed in Section \ref{sec: model} would be given priority in vertical integration choices. Our empirical exercise takes on the firm-level framework by \cite{alfaroetalJPE} and \cite{DelpreteRungi}, while augmenting equations with the Input Rank.

We will make use of a sample of U.S. parent companies that have integrated at least one production stage over time up to 2015. Ownership data and firms' accounts are sourced from the Orbis database, compiled by the Bureau van Dijk. For our scope, we collect information on 20,489 U.S. parent companies controlling 154,836 subsidiaries around the world\footnote{We follow international standards for the identification of corporate control structures \citep{oecd2005, unctad2009, unctad2016}, according to which the unit of observation is the control link between a parent company and each of its subsidiary that is controlled after a concentration of voting rights (50\% plus one stake). Similar data on MNEs have been used in \cite{DelpreteRungi}, \cite{AlviarezCravinoLevchenko} and \cite{CravinoLevchenko2017}}. In Table 4, we provide descriptive statistics of the geographic coverage of the subsidiaries. Both subsidiaries and parent companies can be active in any industry: manufacturing (28.86\%), services (69\%), primary (0.29\%), and extractive (1.85\%). About 81\% of subsidiaries integrated by U.S. parents are domestic. Not surprisingly, U.S. parent companies are engaged mainly in global supply networks across other OECD economies, where 96\% of their subsidiaries are located. The Member States of the European Union host the largest number of foreign subsidiaries in our sample. Among them, Germany, the United Kingdom, and the Netherlands attract a signiﬁcant share of U.S. foreign subsidiaries active in services industries. Not surprisingly, NAFTA members, i.e. Canada and Mexico, mainly host manufacturing of ﬁnal and intermediate goods. However, a non-negligible share of subsidiaries is present in Asia, Africa, and the Middle East.

To validate our sample, we compare with official data on 'Activities of U.S. Multinational Enterprises’ \cite{bea2018}, and against \cite{oecd2018} statistics. In 2015, \cite{bea2018} reports 6,880 billion dollars of total sales by foreign affiliates and 12,628 billion dollars of total sales by parent companies. The U.S. multinational enterprises present in our sample account for 94\% and 92\% of the \cite{bea2018} values, respectively. The number of foreign affiliates in our sample corresponds to 88.6\% on the total of U.S. foreign subsidiaries reported in \cite{oecd2018}, although the latter source only reports the values for the year 2014.

\begin{table}[H]
\centering
\caption{Geographic coverage}
\label{tab:geo coverage}
\resizebox{\textwidth}{!}{%
\begin{tabular}{lcccccccc}
\toprule
\addlinespace
&Final goods & & Intermediate inputs & & Services industries & & All industries&\\ 
Hosting country & N. & \% & N. & \% & N. & \% & N. & \% \\ 
\hline\\
United States & 20,571 & \textit{16.3} & 24,590 & \textit{19.5} & 80.279 & \textit{64.1} & 125,890 & \textit{100.0}\\
\vspace{6pt}
European Union & 1,934 & \textit{11.5} & 2,084 & \textit{12.3} & 12,872 & \textit{76.2} & 16.890 & \textit{100.0}\\
\hspace{12pt}\textit{of which:} & & & & & & & &\\
\hspace{4pt} Germany & 273 & \textit{13.2} & 306 & \textit{14.8} & 1,494 & \textit{72.1} & 2,073 & \textit{100.0}\\
\hspace{4pt} France & 171 & \textit{11.0} & 213 & \textit{13.7} & 1,167 & \textit{75.2} & 1,551 & \textit{100.0}\\
\hspace{4pt} United Kingdom & 563 & \textit{11.4} & 624 & \textit{12.7} & 3,734 & \textit{75.9} & 4,921 & \textit{100.0}\\
\hspace{4pt} Italy & 136 & \textit{19.4} & 139 & \textit{19.8} & 427 & \textit{60.8} & 702 & \textit{100.0}\\
\hspace{4pt} Netherlands & 158 & \textit{6.8} & 171 & \textit{7.3} & 2,005 & \textit{85.9} & 2,334 & \textit{100.0}\\
\\
Canada & 980 & \textit{30.4} & 923 & \textit{28.6} & 1,325 & \textit{41.1} & 3,228 & \textit{100.0}\\
\vspace{6pt}
Russia & 18 & \textit{11.7} & 30 & \textit{19.5} & 106 & \textit{68.8} & 154 & \textit{100.0}\\
Asia & 251 & \textit{15.0} & 312 & \textit{18.7} & 1,109 & \textit{66.3} & 1,672 & \textit{100.0}\\
\hspace{12pt}\textit{of which:} & & & & & & & &\\
\hspace{4pt} Japan & 87 & \textit{11.5} & 76 & \textit{10.1} & 592 & \textit{78.4} & 755 & \textit{100.0}\\
\hspace{4pt} China & 92 & \textit{12.1} & 66 & \textit{8.7} & 605 & \textit{79.3} & 763 & \textit{100.0}\\
\hspace{4pt} India & 122 & \textit{15.7} & 149 & \textit{19.1} & 508 & \textit{65.2} & 779 & \textit{100.0}\\
Africa & 67 & \textit{14.2} & 93 & \textit{19.7} & 313 & \textit{66.2} & 473 & \textit{100.0}\\
\vspace{6pt}
Middle East & 82 & \textit{18.2} & 80 & \textit{17.8} & 288 & \textit{64.0} & 450 & \textit{100.0}\\
\vspace{6pt}
Other Americas & 221 & \textit{12.1} & 495 & \textit{21.6} & 1,210 & \textit{66.3} & 1,926 & \textit{100.0}\\
\hspace{12pt}\textit{of which:} & & & & & & & &\\
\hspace{4pt} Argentina & 24 & \textit{8.1} & 70 & \textit{23.6} & 203 & \textit{68.4} & 297 & \textit{100.0}\\
\hspace{4pt} Brazil & 137 & \textit{14.6} & 219 & \textit{23.3} & 583 & \textit{62.1} & 939 & \textit{100.0}\\
\hspace{4pt} Mexico & 98 & \textit{23.3} & 154 & \textit{36.6} & 169 & \textit{40.1} & 421 & \textit{100.0}\\
\\
Australia & 123 & \textit{14.2} & 157 & \textit{18.1} & 586 & \textit{67.7} & 866 & \textit{100.0}\\
Rest of the world & 489 & \textit{16.5} & 585 & \textit{19.7} & 1,892 & \textit{63.8} & 2,966 & \textit{100.0}\\
\hline
Total & 24,834 & \textit{16.0} & 29,403 & \textit{19.0} & 100,599 & \textit{65.0} & 154,836 & \textit{100.0}\\
\bottomrule
\end{tabular} %
}
\begin{tablenotes}
\footnotesize
\singlespacing
\item Note: firm-level data for year 2015 are sourced from Orbis, by Bureau Van Dijk. A U.S. parent company controls a subsidiary wherever in the world with a direct or indirect equity stake higher than $50\%$, as from international standards. Classification by intermediate and final products is based on correspondence tables with Broad Economic Categories (BEC) rev. 4, as provided by the UN Statistics Division.
\end{tablenotes}
\end{table}

For the scope of our analysis, we map industry affiliations of both parent companies and subsidiaries from the NAICS rev. 2012 classification into the 2002 U.S. BEA I-O Input-Output Tables. The match by industry affiliations allows us combining ﬁrm-level data with sector-level metrics, including the Input Rank we computed as from Section \ref{sec: applications}, the Relative Upstreamness segments sourced from \cite{alfaroetalJPE}, and the direct requirement coefficient originally retrieved from I-O tables. All three previous indicators are industry-pair specific, i.e., considering the industry of the parent company and the one of each subsidiary. The underlying assumption is that the affiliation to the parent company indicates the buying industry, and the affiliation of each subsidiary indicates the (direct or indirect) supplying industry. In the absence of actual data on firm-to-firm transactions, such a mapping allows us proxying buyer-supplier relationships. For more details, see Section \ref{sec: applications}. Finally, we complement our data with input-industry estimates of demand elasticity sourced from \cite{BrodaWeinstein2006}. In Appendix Table \ref{tab: stats}, we report descriptive statistics of industry-level variables after the matching with sample firms and baseline estimates.

\subsection{Empirical strategy}
Let $h=1, 2, … , N$ denote the set of inputs, as from the I-O tables, and let $r=1, 2, … , R$ denote the set of parent companies, each active in an output industry, $k=1, 2, …, K$. The dependent variable, $y_{r(k)}$, takes on a value 1 when at least one subsidiary in the $h$-th input market has been integrated by a parent $r$ in industry $k$, and $0$ otherwise. Therefore, for each parent company, we have a vector $\boldsymbol y_{r(k)} = \left[ y_{1r(k)}, … , y_{Nr(k)}\right]$ made of $0$s and $1$s when an $h$-th input has been integrated or not, respectively. At this point, we can consider the probability that a generic parent chooses among a set of alternatives such that:

\begin{align}
Pr\left(\boldsymbol y_{r(k)}|\sum_{h=1}^N y_{hr(k)}\right) = \frac{exp\left[ y_{hr(k)}\boldsymbol x_{hr(k)} \boldsymbol{\beta}\right]}{\sum_{\boldsymbol s_{r(k)} \in \mathfrak{S}_{r(k)}}^N exp\left[\boldsymbol s_{r(k)}\boldsymbol x_{hr(k)} \boldsymbol{\beta}\right]}
\end{align}

where the element $s_{hr(k)}$ of a vector $\boldsymbol s_{r(k)}$ is equal to $1$ if an input has been integrated, and zero otherwise. Thus, $\boldsymbol s_{r(k)}$ indicates the combination of inputs that have been integrated by the $r$-th parent company after considering all the possible combinations one could pick from $\mathfrak{S}_{r(k)}$. Therefore, we identify a vector of covariates for each input-output pair, $x_hr(k)$, which includes: i) the Input Rank of the $h$-th input with respect to the $k$-th output estimated with a dampening factor equal $\chi = 1$, i.e., assuming each parent company perfectly knows the entire production network; ii) a binary variable Complements relative to the $k$-th output market; the Upstreamness sourced from \cite{alfaroetalJPE}; iii) the bilateral direct requirement coefficient sourced directly from I-O tables. As in \cite{AntrasChor} and \cite{alfaroetalJPE}, the variable Complements is equal to 1 when the elasticity of substitution of the output market is below the median ($\rho_k > \rho_{med}$), and 0 otherwise. Errors are clustered by input markets. Fixed effects are included at the parent level. Results from nested specifications are reported in Tables \ref{tab: baseline}.

\subsection{Results}

The odds ratios for the Input Rank are always higher than one and significant, therefore parent companies preferably integrate (direct or indirect) inputs that are more central in their production network because a productivity shock in those upstream markets has the potential to hit harder on their downstream production processes. The magnitudes of the odds ratios are quite stable across baseline specifications, in a range between $1.25$ and $1.41$. They are also quite stable after controlling for Upstreamness and direct requirement coefficients.

Interestingly, odds ratios lower than one on Upstreamness indicate that it is less likely that a parent company integrates distant inputs in the production network, i.e., the more upstream they are the less likely they are integrated. From this point of view, it is clear enough that the Input Rank and the Upstreamness convey different information on positions along GVCs. In general, the coefficients on direct requirements are not robust along different specifications that separate parent companies producing final goods from the ones producing intermediate inputs.

\begin{table}[H]
\centering
\onehalfspacing
\caption{Input Rank and vertical integration - Fixed-effects conditional logit}
\label{tab: baseline}
\resizebox{0.9\textwidth}{!}{%
\begin{tabular}{lccccc}
\toprule
\addlinespace
Dependent variable:  & (1) & (2) & (3) & (4) & (5)\\
input is integrated = 1 & & & & &\\
\addlinespace
\hline
$Input\:Rank_{hk}$ & 1.413*** & 1.348*** & 1.276*** & 1.245*** & 1.354**\\
& (.093) & (.066) & (.098) & (.104) & (.142)\\
$Input\:Elasticity\left[\epsilon>med\right]_{h}$ &  &  & .554*** & .583* & .523**\\
&  &  & (.120)  & (.186)  & (.160)\\
$Input\:Rank_{hk} \times Input\:Elasticity\left[\epsilon > med \right]_{h}$ &  &  & 1.179* & 1.200 & 1.135\\
 &  &  & (.108)  & (.148) & (.162)\\
$Upstreamness_{hk}$ &  & .510*** & .540*** & .567*** & .512***\\
 &  & (.057)  & (.071) & (.098) & (.105)\\
$Direct\:Requirement_{hk}$ &  & 1.047** & 1.031*** & 1.032 & 1.033\\
 &  & (.027) & (.018) & (.031) & (.021)\\
\hline
N. obs. & 7,805,667 & 7,805,667 & 1,131,406 & 586,782 & 531,002\\
N. Parent companies  & 20,489 & 20,489 & 4,069 & 2,110 & 1,910\\
Parent fixed effects  & Yes & Yes & Yes & Yes & Yes\\
Clustered errors  & Yes & Yes & Yes & Yes & Yes\\
Activity of parent companies & All & All & All & Final & Intermediate\\
Pseudo R-squared (McFadden's)  & .402 & .415 & .226 & .174 & .279\\
Log pseudo-likelihood & -94,000.003 & -92,012.326 & -23,417.893 & -13,147.021 & -10,120.049\\
\addlinespace
\hline
\bottomrule
\end{tabular} %
}
\begin{tablenotes}
\footnotesize
\singlespacing
\item Note: Odds ratios after a parent-level fixed-effects conditional logit are reported in the form $\frac{Pr(y_{hr(k)}=1)}{Pr(y_{hr(k)}=0)}$. Errors are clustered by I-O output industries. Variables are standardized at their mean and variance. *, **, *** stand for $p < .10$, $p < .05$, and $p < .01$, respectively. The variable $Input\:Elasticity\left[\epsilon>med \right]_{h}$ is available only for traded inputs.
\end{tablenotes}
\end{table}

Notably, we also find that parent companies less likely integrate suppliers from highly elastic industries, i.e., from suppliers that produce inputs whose substitutes are easier to find on the market. Please note that the main reason why the sample reduces from column 2 to column 3 is because input elasticities are available only for traded inputs, as they are sourced from an exercise on U.S. international trade data \citep{BrodaWeinstein2006}.

Results are robust to controls for sample composition in Table \ref{tab: sensitivity}, when we look only at manufacturing inputs (columns 1 and 2) or manufacturing outputs (columns 3 and 4), as well as when we look only at top 100 inputs with the highest direct requirement coefficients (columns 5 and 6). Finally, main findings are not sensitive to changes in the functional form. In Appendix Table \ref{tab: functional forms}, we test alternatively that simple probit, logit, and linear probability models present similar correlations between vertical integration choices, the Input Rank, the Upstreamness, and the input elasticity.

\begin{table}[H]
\centering
\onehalfspacing
\caption{Input Rank and vertical integration - Robustness to sample composition}
\label{tab: sensitivity}
\resizebox{0.9\textwidth}{!}{%
\begin{tabular}{lcccccc}
\toprule
\addlinespace
Dependent variable:  & Manuf input & Manuf input & Manuf output & Manuf output & Top 100 & Top 100 \\
input is integrated = 1 & & & & & &\\
\addlinespace
\hline
$Input\:Rank_{hk}$ & 1.357*** & 1.306*** & 1.369*** & 1.315*** & 1.391*** & 1.275***\\
& (.066) & (.083) & (.065) & (.082) & (.071) & (.082)\\
$Input\:Elasticity\left[\epsilon>med \right]_{h}$ &  & .709* &  & .577** & & .328***\\
&  & (.067) &  & (.126) & & (.109) \\
$Input\:Rank_{hk}\: \times \: Input\:Elasticity \left[\epsilon>med\right]_{h}$ &  & 1.145  &  & 1.158 & & 1.324**\\
 &  & (.108) & & (.106) & & (.145) \\
$Upstreamness_{hk}$ & .496*** & .532*** & .553*** & .762 & .751 & .812\\
 & (.070) & (.075)  & (.072) & (.143) & (.134) & (.178)\\
$Direct\:Requirement_{hk}$ & 1.009  & 1.022 & 1.044** & 1.028 & 1.062*** & 1.034*\\
 &  (.035) & (.025) & (.020) & (.071) & (.019) & (.018) \\
\hline
N. obs. & 935,648 & 893,464 & 1,396,382 & 1,103,884 & 254,201 & 154,824\\
N. Parent companies  & 2,175 & 2,134 & 4,166 & 3,970 & 256 & 240 \\
Parent fixed effects  & Yes & Yes & Yes & Yes & Yes & Yes\\
Clustered errors  & Yes & Yes & Yes & Yes & Yes & Yes\\
Pseudo R-squared (McFadden's)  & .239 & .248 & .197 & .229 & .267 & .350\\
Log pseudo-likelihood & -20,124.371 & -19,339.866 & -30,250.643 & -22,800.165 & -13,785.276 & -8,522.782\\
\addlinespace
\hline
\bottomrule
\end{tabular}%
}
\begin{tablenotes}
\footnotesize
\singlespacing
\item Note: Odds ratios after a parent-level fixed effects conditional logit are reported in the form $\frac{Pr(y_{hr(k)}=1)}{Pr(int_{hr(k)}=0)}$. Errors are clustered by I-O output industries. Variables are standardized at their mean and variance. *, **, *** stand for $p < .10$, $p < .05$, and $p < .01$, respectively. Please note, the variable $Input\:Elasticity\left[\epsilon>med\:\right]_{h}$ is available only for traded industries.
\end{tablenotes}
\end{table}

Please note, once again, that results reported in Tables \ref{tab: baseline}, \ref{tab: sensitivity}, and \ref{tab: functional forms} cannot be interpreted in a structural way, because we do not provide for a full incentive structure of the relations established between buyers and suppliers in a production network. For example, we do not consider the contractual environment that can endogenously determine the positions of suppliers and the decisions to keep them at arms' length. Yet, our intention is to show that network dimensions in GVCs matter, and our correlations point to a significant role for input centrality, in the sense we introduced in Section \ref{sec: model}, once we consider the possible impact of a productivity shock from the upstream market.

\section{Conclusions} \label{sec: conclusions}

In this contribution, we introduced the Input Rank as a theoretically sound measure to catch the relevance of both direct and indirect suppliers from the perspective of downstream buyers, whose production process can rely on increasingly sophisticated supply networks in times of a global fragmentation. 
In our framework, a supplier ranks higher for a downstream buyer if a productivity shock has a higher impact on that buyer's marginal costs after transmission through the supply network. Besides a full account of the topology of a production network, we allow for a rich heterogeneity concerning how much firms rely on intermediate inputs, as well as on how strong competitive forces are within each sector. At each stage of production, a lower intensity in the usage of intermediate inputs (i.e., a higher intensity in labor services) can buffer the transmission of a shock from upstream markets. At the same time, a higher markup (i.e., a lower competitive pressure) on an upstream market will make downstream buyers more sensitive to input-specific productivity shocks. 

When we bring our model to I-O tables, we discover how central services industries are in production networks, as they connect across most manufacturing and services industries in a modern economy. After we look at international I-O tables, we retrieve a pecking order in the geography of sourcing. Most crucial inputs come from domestic producers, in the first place, and from regionally integrated partners, in the second place.

Finally, we show that the Input Rank is a good predictor of vertical integration choices on a sample of U.S. parent companies controlling subsidiaries on a global scale, possibly because upstream shocks can be mitigated after vertical integration.

Further work is needed to address, among others, the endogenous relationship between the organization of the global supply network and the contractual environment. However, in general, we argue that the Input Rank enables to catch the recursive and complex nature of real-world supply networks, which are too often represented as supposedly linear chains in studies on the international organization of production. Ours is a first step that shows that technological non-linearities blend with endogenous market forces to shape global production networks, whose loops, kinks, and corners can magnify or dampen policies and shocks on globally integrated markets.

\bibliographystyle{abbrvnat}
\bibliographystyle{apalike}
\newpage
\bibliography{sample.bib}

\setcounter{equation}{0}
\renewcommand{\theequation}{A\arabic{equation}}
\begin{center}
    
\section*{Appendix A: Proofs}
\bigskip
\end{center}
\doublespacing

\begin{lemma} \label{lem:CostFunction}
Let $N_k^+$ denote the set of inputs of firms in sector k. The cost function of firm $i$ is given by
\begin{align*}
c(y(k,i); w, \mathbf{p})= \lambda(k,i) y(k,i),
\end{align*}

where $\lambda(k,i) = w^{\beta_k}\prod_{h\in N_k^+}p_h^{\alpha_k g_{hk}}\tau_h^{\alpha_k g_{hk}}.$
\end{lemma}	
\bigskip

\begin{proof}[\textbf{Proof of Lemma \ref{lem:CostFunction}}]

The Langragian of the cost minimization  problem of firm $i$ from sector $k$:
\small{
\begin{align*}
\mathcal{L}=  w\ell(k,i) +  \sum_{h=1}^{M} p_h x(k,i,h)- \lambda(k,i) \left[ \zeta_k \ell(k,i)^{\beta_{k}} \left(\prod_{h=1}^{M} (\tau_h^{-1} x(k,i,h))^{g_{hk}}\right)^{\alpha_{k}} - y(k,i) \right]
\end{align*}}

From the first-order necessary conditions (also sufficient, given convexity), we can  deliver the following conditional demand functions:
\begin{align} \label{eq:CondDemand}
\begin{split}
x(k,i,h)= \lambda(k,i)  \alpha_k g_{hk}\frac{y(k,i)}{p_h}; \hspace{8 pt} \ell(k,i) = \lambda(k,i) \beta_k \frac{y(k,i)}{w}.
\end{split}
\end{align}

Plugging \eqref{eq:CondDemand} into the expression for the cost of firm $i$, it directly follows that $c(y(k,i), w \mathbf{p}) = \lambda(k,i)y(k,i)$. Hence $\lambda(k,i)$ is the marginal cost of production of firm $i$. 

Substituting  \eqref{eq:CondDemand} into the production function we obtain:

\begin{align*}
y(k,i) = &\zeta_k  \left(\frac{\lambda(k,i) \beta_k y(k,i)}{w}\right)^{\beta_k} \prod_{h \in N_k^+}\left(\tau_h^{-1}\frac{\lambda(k,i) \alpha_k g_{hk}y(k,i)}{p_h}\right)^{g_{hk}\alpha_k}\\
= &\zeta_{k}\lambda(k,i) y(k,i)\left(\frac{\beta_k}{w}\right)^{\beta_k}\prod_{h \in N_k^+}\left(\frac{g_{hk}\alpha_k}{ \tau_h p_h} \right)^{g_{hk}\alpha_k}\\
= &\lambda(k,i) w^{-\beta_k}\prod_{h \in N_k^+}\left(p_h \tau_h\right)^{-g_{hk}\alpha_k}y(k,i),
\end{align*}
where for the last equality we used the fact that $\zeta_k = \beta_k^{-\beta_k}\prod_{h \in N_k^+}(\alpha_k g_{hk})^{-\alpha_k g_{hk}}$. 

Solving for $\lambda(k,i)$, we get:
\begin{align}\label{eq:MarginalCostExpression}
\lambda(k,i) = w^{\beta_k}\prod_{h \in N_k^+}p_h^{\alpha_k g_{hk}}\tau_h^{-\alpha_k g_{hk}}.
\end{align}

as desired.
\end{proof}
\bigskip

\begin{lemma}\label{lem:Markup}
The following relations between firm-level marginal cost $\lambda(k,i)$ and sector-level marginal cost $\lambda_k$, and firm level markup $\mu(k,i)$ and sector level markup $\mu_k$ hold:
\begin{align*}
\lambda_k = M_k^{\frac{1}{1-\epsilon_k}}\lambda(k,i),\\
\mu_k = \mu(k,i) =\frac{\epsilon_k}{\epsilon_k -1}
\end{align*}
\end{lemma}

\begin{proof}[\textbf{Proof of Lemma} \ref{lem:Markup}]
Using results from the theory of monopolistic competition, the sector-level price of good $k$, $p_k$, and the sector-level output, $y_k$, are given by:
\begin{align}\label{eq:SectorPriceAndOutput}
\begin{split}
p_k = \left(\sum_{i=1}^{M_k}p(k,i)^{1- \epsilon_k}\right)^{\frac{1}{1-\epsilon_k}},\\
y_{k} =  \left(\sum_{i=1}^{M_k}y(k,i)^{\frac{\epsilon_k-1}{\epsilon_k}}\right)^{\frac{\epsilon_k}{\epsilon_k-1}}.
\end{split}
\end{align}
  
In the symmetric equilibrium, $p(k,i) = p(k,j)$, and $y(k,i) = y(k,j)$. Hence:
\begin{align*}
p_k = M_k^{\frac{1}{1-\epsilon_k}}p(k,i),\\
y_k = M_k^{\frac{\epsilon_k}{\epsilon_k-1}}y(k,i).
\end{align*}
Thanks to the assumption on constant returns to scale, and using the expression for $y_k$, we can write the sector-level marginal cost of production as:
\begin{align*}
\lambda_k = \sum_{i=1}^{M_k}\lambda(k,i)\frac{y(k,i)}{y_k} = \frac{M_k \lambda(k,i) y(k,i)}{M_k^{\frac{\epsilon_k}{\epsilon_k-1}}y(k,i)} = M_k^{\frac{1}{1-\epsilon_k}}\lambda(k,i)
\end{align*} 
  
Finally, from the firm pricing rule we have $p(k,i) = \mu(k,i)\lambda(k,i) $. Plugging the pricing rule in the expression for $p_k$ from \eqref{eq:SectorPriceAndOutput} we get:
\begin{align*}
p_k =  &\left(\sum_{i=1}^{M_k}\left(\mu(k,i) \lambda(k,i)\right)^{1- \epsilon_k}\right)^{\frac{1}{1-\epsilon_k}} \\
= &\left[M_k \mu(k,i)^{1-\epsilon_k}\lambda(k,i)^{1-\epsilon_k}\right]^{\frac{1}{1-\epsilon_k}}=M_k^{\frac{1}{1-\epsilon_k}}\mu(k,i)\lambda(k,i)=\mu(k,i)\lambda_k
\end{align*}
 which completes the proof.
\end{proof}
\bigskip

\begin{lemma}\label{lem:AggregationDemand}
The aggregate demand of sector $k$ for deliveries of input $h$ is given by: 
\begin{align}\label{eq:SectorLevelDemand}
x(k,h) = \frac{\alpha_k g_{hk}p_ky_k}{\mu_k p_h} =  \frac{\epsilon_k-1}{\epsilon_k}\frac{\alpha_k g_{hk}p_ky_k}{p_h}.
\end{align}
\end{lemma}

\begin{proof}[\textbf{Proof of Lemma} \ref{lem:AggregationDemand}]
Symmetry implies that $x(k,h) = M_{k}x(k,i,h)$. From \eqref{eq:CondDemand} and Lemma \ref{lem:Markup}, we get:
\begin{align*}
x(k,h) = &M_k \lambda(k,i) \alpha_k g_{hk}\frac{y(k,i)}{p_h}
 = M_k M_k^{-\frac{1}{1-\epsilon_k}}\frac{p_k}{\mu_k}\alpha_k g_{hk}M_k^{-\frac{\epsilon_k}{\epsilon_k-1}}y_k =\frac{\alpha_k g_{hk}p_ky_k}{\mu_k p_h}
\end{align*}
where we used the fact that $p_k =M_k^{\frac{1}{1-\epsilon_k}}\mu(k,i)\lambda(k,i) $ and that $\mu_k = \mu(k,i)$ for $i \in k$.
\end{proof}
\bigskip
\begin{proof}[\textbf{Proof of Proposition \ref{prop:PropositionMarginalCostPartial}}]
For notational simplicity, let  $\check{x}$  denote $\log x$  for any variable $x$. Taking logarithms  of \eqref{eq:MarginalCostExpression} and using the pricing rule ($p_h = \mu_h \lambda_h $), we get:
\begin{align*}
\check \lambda(k,i)=  \beta_k \check{w} + \alpha_k \sum_{h \in N_K^+} g_{hk} \left( \check \mu_h +\check \tau_h \right) + \alpha_k \sum_{h \in N_k^+}g_{hk} \check \lambda_h.
\end{align*}
From Lemma \ref{lem:Markup}, we can write the above equation in terms of sector level marginal cost as:
\begin{align*}
\check \lambda_k = \frac{1}{\epsilon_k-1} \check{M}_k + \beta_k \check w + \alpha_k \sum_{h \in N_k^+} g_{hk} \left(\check \mu_h +\check \tau_h \right) + \alpha_k \sum_{h \in N_k^+}g_{hk} \check \lambda_h
\end{align*}
Define column vector $\vec{\eta} = \{\frac{1}{\epsilon_k-1} \check{M}_k \}_{k=1}^M $. Moreover normalize $w=1$.  Writing the above equation for all $k$'s in vector notation we get:
\begin{align}\label{eq:CostCentrality}
\begin{split}
\boldsymbol{\check{\lambda}}  =& \vec{\eta}  +\mathbf{AG'}\check{\boldsymbol{\mu}} +\mathbf{AG'} \boldsymbol{\tau} + \mathbf{AG'} \boldsymbol{\lambda} \Rightarrow\\
\boldsymbol{\check{\lambda}} =& \left[\mathbf I - \mathbf{AG'}\right]^{-1}( \vec{\eta}+ \mathbf{AG'}(\check{\boldsymbol{\mu}} +\check{\boldsymbol{\tau}})).
\end{split}
\end{align}
%
Finally, by differentiating  we get:
\begin{align*}
\frac{\partial\check\lambda(k,i)}{\partial \check \tau_h}= &\frac{\partial \check\lambda_k}{\partial \check\tau_h} =  \boldsymbol{e}_k'\left[\mathbf I - \mathbf{AG'}\right]^{-1}\mathbf{AG'} \boldsymbol{e}_h =  \boldsymbol{e}_k' \left[\left[\mathbf I - \mathbf{AG'}\right]^{-1} - \mathbf I \right] \boldsymbol{e}_h\\
= &\boldsymbol{e}_h'\left[\left[\mathbf I -  \mathbf{GA}\right]^{-1} - \mathbf I \right] \boldsymbol{e}_k.
\end{align*}
Whenever $k \neq h$ the above equation has a form:
\begin{align*}
\frac{\partial \check \lambda(k,i)}{\partial \check \tau_h}  = \boldsymbol{e}_h' \left[\mathbf I - \mathbf{GA} \right]^{-1} \boldsymbol{e}_k,
\end{align*}
which completes the proof. 
\end{proof}



\begin{proof}[Proof of Proposition \ref{prop:CompStatics}]

From \eqref{eq:CostCentrality} we have:
\begin{align*}
    \check \lambda_k = \sum_{\ell=1}^{M}\upsilon_{\ell k } \eta_{\ell} + \sum_{\ell=1}^{M}\upsilon_{ \ell k } \check{\mu}_{\ell}+ \sum_{\ell=1}^{M}\upsilon_{\ell k } \check{\tau}_{\ell} - \check{\mu}_k - \check{\tau}_k,
\end{align*}
where $\upsilon_{\ell k}$ is the Input Rank of input $\ell$ for firms in $k$. Using Lemma \ref{lem:Markup} we get:
\begin{align*}
    \check{\lambda}(k,i) = \sum_{\ell=1}^{M}\upsilon_{\ell k } \eta_{\ell} + \sum_{\ell=1}^{M}\upsilon_{ \ell k } \check{\mu}_{\ell}+ \sum_{\ell=1}^{M}\upsilon_{\ell k } \check{\tau}_{\ell} - \check{\mu}_k - \check{\tau}_k - \eta_k,
\end{align*}
which gives:
\begin{align*}
    \lambda(k,i) = \exp \left(\sum_{\ell=1}^{M}\upsilon_{\ell k } \eta_{\ell} + \sum_{\ell=1}^{M}\upsilon_{ \ell k } \check{\mu}_{\ell}+ \sum_{\ell=1}^{M}\upsilon_{\ell k } \check{\tau}_{\ell} - \check{\mu}_k - \check{\tau}_k - \eta_k \right)
\end{align*}
where $\exp$ denotes the exponential function. Define:
\begin{align*}
    \delta_k \equiv \sum_{\ell=1}^{M}\upsilon_{\ell k } \eta_{\ell} + \sum_{\ell=1}^{M}\upsilon_{ \ell k } \check{\mu}_{\ell}+ \sum_{\ell=1}^{M}\upsilon_{\ell k } \check{\tau}_{\ell} - \check{\mu}_k - \check{\tau}_k - \eta_k.
\end{align*}
By simple differentiation of $\lambda(k,i)$ we get:
\begin{align*}
    \frac{\partial \lambda(k,i)} {\partial \tau_h}=  \frac{\upsilon_{h k}}{\tau_h} e^{\delta_k},
\end{align*}
from where if follows  directly that $\frac{\partial^2 \lambda(k,i)} {\partial \tau_h \partial \upsilon_{hk}} >0$.

To calculate $\frac{\partial^2 \lambda(k,i)} {\partial \tau_h \partial \epsilon_h}$  we recall that $\mu_h = \frac{\epsilon_h}{\epsilon_h -1}$  and $\eta_k = \frac{1}{\epsilon_k-1}\check{M}_k$, an thus:
\begin{align*}
    \frac{\partial^2 \lambda(k,i)} {\partial \tau_h \partial \epsilon_h}=  \frac{\upsilon_{h k}^2}{\tau_h}e^{\delta_k} \left(\frac{\partial \eta_h}{\partial \epsilon_h} +\frac{\partial \check{\mu}_h}{\partial \epsilon_h} \right)  = \frac{\upsilon_{h k}^2}{\tau_h}e^{\delta_k} \frac{1-\epsilon_h(1+\check{M}_h)}{(\epsilon_h-1)^2 \epsilon_h} <0,
\end{align*}
where the last inequality comes from the fact that $M_h \geq 1$ and $\epsilon_h >1$, which completes the proof.


\end{proof}
\newpage
\setcounter{figure}{0}
\renewcommand{\thefigure}{B\arabic{figure}}
\begin{center}
\section*{Appendix B:\\
Imperfect information on the supply network}\label{sec:IncompleteInformation}
\end{center}

In this Appendix, we consider the case of a producer that does not fully observe her supply network. Intuitively, a firm cannot outreach the entire web of specialized suppliers when production networks are more sophisticated. To fix ideas, it is worth looking back at the fictional supply network reported in Figure \ref{fig: fictional}. Any time the manager of a firm 1 tries to collect information about upstream transactions, say about transactions between firm 3 and its suppliers, she has a limited ability to know the quality and quantity of deliveries. She can call the direct supplier to ask or, alternatively, she can gather information on the market when, for example, prices and qualities of upstream inputs are relatively standard terms to include in a contract. At any upstream step, e.g. from firm 4 up to firm 6, the same problem starts all over again. 

We capture this phenomena in a stylized way by assuming that firm $i$ in a sector $k$ does not observe the full network $\mathbf{G}$ but rather its subnetwork $\chi_k \mathbf{G}$, where one can think of $\chi_k$ as a parameter capturing the share of transactions in the production network observed by firms form sector $k$. We assume that the extent to which the network is observed ($\chi_k$) is output-specific, i.e., it varies according to the peculiar contractual environment of the downstream market. In this, we follow the seminal intuition by \cite{rauch1999}, who sketched the idea of a network search on international trade when products are differentiated or homogeneous. In line with that intuition, we can think of $\chi_k$ as a search barrier in supply networks. Even more realistically, one may consider an extension to the case when firms in sector $k$ observe suppliers of firms in sector $h$ with independent probabilities $\chi_{hk}$. Then, we would replace a scalar $\chi_k$ with a diagonal matrix $\mathbf{H}_k$ that has diagonal elements equal to $\chi_{hk}$. Taking the output as a reference, we align with \cite{Nunn2007}, who looks at an average measure of input contract intensity to infer the 'thickness' and the relation-specificity of the markets. 

Hence, when assessing the importance of suppliers from sector $h$, instead of relying on Input Rank $\nu_{hk}(\mathbf{GA}) = \mathbf{e}_h' \mathbf{(I- GA)} \mathbf{e}_k$, firms from sector $k$ consider  Input Rank  $\nu_{hk}(\chi_k \mathbf{GA}) = \mathbf{e}_h' \mathbf{(I-  \chi_k GA)^{-1}} \mathbf{e}_k$. Since $\chi_k \leq 1$, firm $i \in k$ underestimates the importance of indirect suppliers. When $\chi_k$ is smaller, suppliers that are relatively closer to the final producer will have a relatively higher Input Rank than more distant suppliers. To illustrate this property, we calculate the Input Rank using the network from Figure \ref{fig: fictional} at changing values of the dampening rates $\chi$ (assumed to be equal for each node $k$) and plot it in Figure \ref{fig: comparative statics}.
%
\\
\begin{figure}[H]
\centering
\caption{Input Rank as a function the searching parameter $\chi$}\label{fig: comparative statics}
\includegraphics[width=0.65\textwidth]{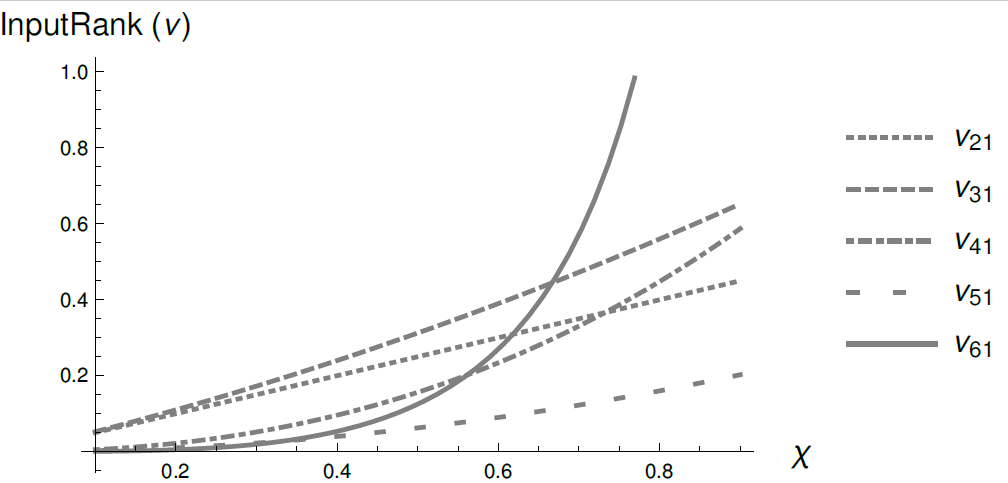}
\begin{tablenotes}
\footnotesize
\singlespacing
\item Note: Input Ranks for the output of firm 1 calculated for a fictional supply network depicted in Figure \ref{fig: fictional}, assuming symmetric input deliveries and labor intensities, as functions of $\chi$. 
\end{tablenotes}
\end{figure}

When plotting Figure \ref{fig: comparative statics}, we assume that for any fixed output node $k$ and any two of its suppliers $h_1$ and $h_2$, we have $g_{h_{1}k} = g_{h_{2}k}$. Interestingly, although there are more (upstream) paths connecting firm 1 to firm 6 than firm 1 to firm 2, firm 2 will have a disproportionately higher Input Rank when the dampening rate becomes smaller and smaller. In a nutshell, a limited ability to outreach on the input markets implies that downstream buyers underestimate the role of longer paths in production network. In other words, search barriers can prevent the downstream buyers to explore production processes that are too distant in the supply structure.

\newpage
\setcounter{table}{0}
\renewcommand{\thetable}{C\arabic{table}}
\centering
\setcounter{figure}{0}
\renewcommand{\thefigure}{C\arabic{figure}}

\section*{Appendix C: Tables and Graphs} \label{sec:tab}
\vspace{2cm}

\begin{figure}[H]
\centering
\caption{In-degree distribution of Input-Output Network from U.S. BEA 2002 I-O tables}
\label{fig: indegree US}
\includegraphics[width=0.7\textwidth]{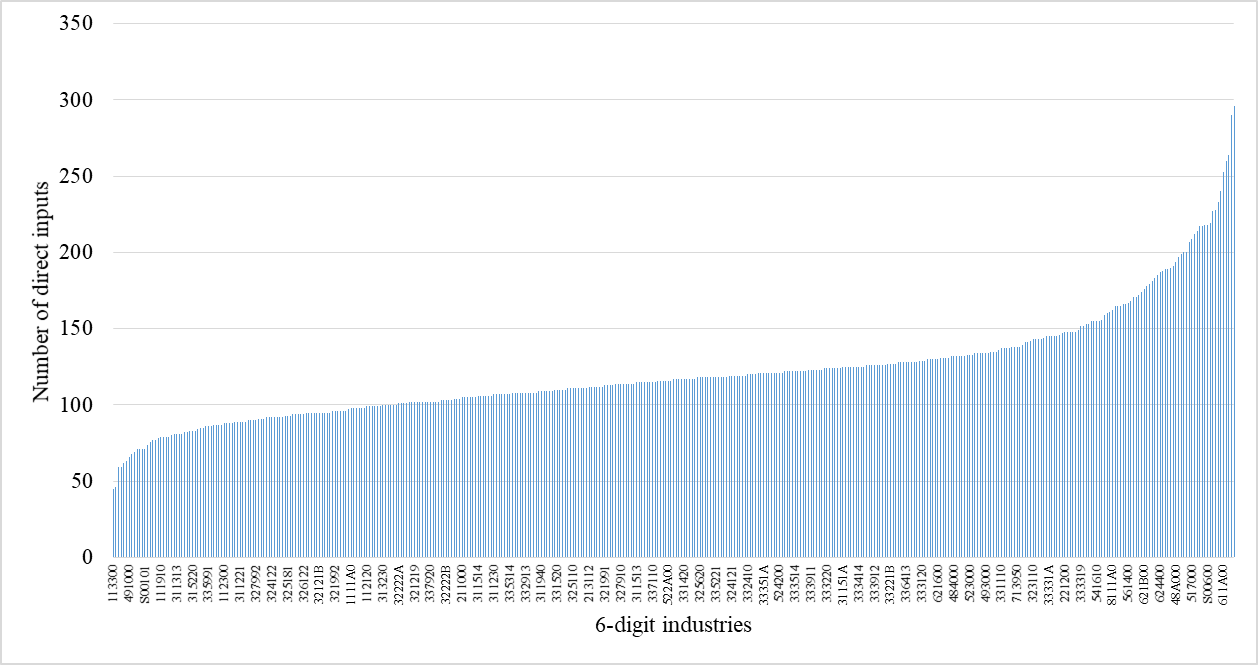}
\begin{tablenotes}
\footnotesize
\singlespacing
\item Note: Number of input industries by output ordered on the x-axis. Average: 122. Minimum at the Logging industry (code 113300) is 45. Maximum at the Retail Trade (code 4A0000) is 296.
\end{tablenotes}
\end{figure}
\bigskip

\begin{figure}[H]
\begin{center}
\caption{Out-degree distribution of Input-Output Network from U.S. BEA 2002 I-O tables}
\label{fig: outdegree US}
\includegraphics[width=0.7\textwidth]{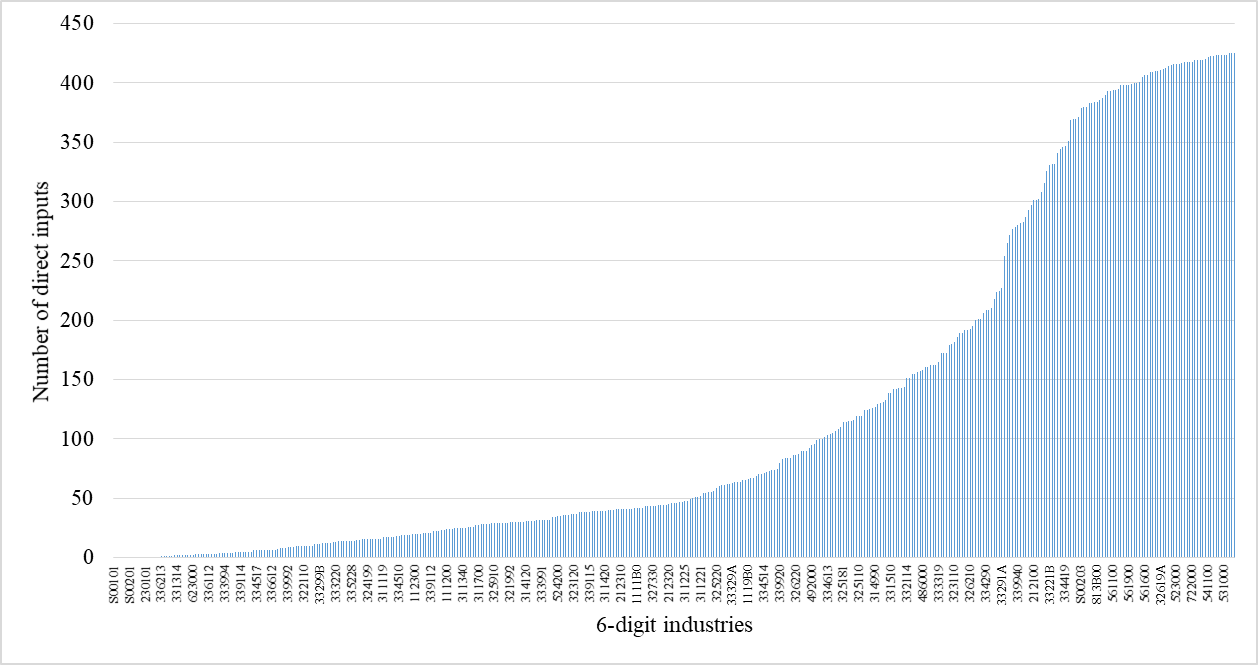}
\end{center}
\begin{tablenotes}
\footnotesize
\singlespacing
\item Note: Number of buying industries by output ordered on the x-axis. Average: 122. Minimum at the Museums, Historical Sites, Zoos, and Parks (code 712000) is 0. Maximum at the Wholesale Trade (code 420000) is 425.
\end{tablenotes}
\end{figure}
\newpage
\begin{table}[H]
\centering
\caption{Top 10 inputs of the Automobile Manufacturing (I-O code: 336111)}
\label{tab: top10 auto}
\resizebox{\textwidth}{!}{%
\begin{tabular}{cllc}
\toprule
\addlinespace
Abs. Rank & I-O code & Industry name & Input Rank\\
\hline
1 & 336111 & Automobile Manufacturing & .10263\\
2 & 336300 & Motor Vehicle Parts Manufacturing & .05608\\
3 & 420000 & Wholesale Trade & .01183\\
4 & 550000 & Management of Companies and Enterprises & .01103\\
5 & 531000 & Real Estate & .00965\\
6 & 331110 & Iron and Steel Mills and Ferroalloy Manufacturing & .00834\\
7 & 211000 & Oil and Gas Extraction & .00721\\
8 & 533000 & Lessors of Nonfinancial Intangible Assets & .00521\\
9 & 221100 & Electric Power Generation, Transmission, and Distribution & .00492\\
10 & 324110 & Petroleum Refineries & .00492\\
\bottomrule
\end{tabular}}
\begin{tablenotes}
\footnotesize
\singlespacing
\item Note: The full Input Rank vector of the Automobile Manufacturing (code 336111) has been computed starting from all 421 industries classified at the 6-digit in the U.S. BEA 2002 tables.
\end{tablenotes}
\end{table}

\begin{table}[H]
\centering
\caption{Top 10 inputs of the Electronic Computer Manufacturing (I-O code: 334111)}
\label{tab: top10 computer}
\resizebox{\textwidth}{!}{%
\begin{tabular}{cllc}
\toprule
\addlinespace
Abs. Rank & I-O code & Industry name & Input Rank\\
\hline
1 & 334111 & Electronic Computer Manufacturing & .09650\\
2 & 334112 & Computer Storage Device Manufacturing & .01882\\
3 & 334418 & Printed Circuit Assembly Manufacturing & .01472\\
4 & 420000 & Wholesale Trade & .01354\\
5 & 334413 & Semiconductor and Related Device Manufacturing & .01318\\
6 & 550000 & Management of Companies and Enterprises & .01220\\
7 & 531000 & Real Estate & .00981\\
8 & 511200 & Software Publishers & .00675\\
9 & 33411A & Computer Terminals and Other Computer Peripheral Equipment & .00626\\
10 & 541800 & Advertising and Related Services & .00459\\
\bottomrule
\end{tabular}}
\begin{tablenotes}
\footnotesize
\singlespacing
\item Note: The full Input Rank vector of the Automobile Manufacturing (code 336111) has been computed starting from all 421 industries classified at the 6-digit in the U.S. BEA 2002 tables.
\end{tablenotes}
\end{table}

\begin{table}[H]
\centering
\caption{Descriptive statistics of industry-level variables}
\label{tab: stats}
\resizebox{.7\textwidth}{!}{%
\begin{tabular}{lccccc}
\toprule
\addlinespace
Variable & Mean & St. Dev. & Min. & Max & N. obs.\\
\hline 
Input Rank & .0021 & .0183 & .0001 & .5530  & 1,131,406 \\
Upstreamness & 3.2752 & .8865 & 1.0029 & 8.7470 & 1,131,406\\
Direct requirement coeff. & .0009 & .0073  & 0 & .4194 & 1,131,406 \\
Elasticity of substitution & 8.7239 & 11.8644 & 1.3 & 108.5019 & 1,131,406\\
\bottomrule
\end{tabular}}
\begin{tablenotes}
\footnotesize
\singlespacing
\item Note: Input Rank, Upstreamness, and direct requirement coefficients are based on U.S. I-O 2002 Tables sourced from \cite{usbea}. Elasticities of substitution for inputs are sourced from \cite{BrodaWeinstein2006}. Number of sample observations refer to the last columnn of baseline estimates in Table \ref{tab: baseline}.
\end{tablenotes}
\end{table}

\begin{table}[H]
\centering
\onehalfspacing
\caption{Input Rank and vertical integration - Different functional forms}
\label{tab: functional forms}
\resizebox{0.99\textwidth}{!}{%
\begin{tabular}{lccccccc}
\toprule
\addlinespace
Dependent variable:  & LPM & LPM & Logit & Logit & Probit & Probit \\
input is integrated = 1 & & & & & &\\
\addlinespace
\hline
$Input\:Rank_{hk}$ & .027*** & .023*** & 1.324*** & 1.266*** & .140*** & .127***\\
& (.006) & (.008) & (.065) & (.090) & (.035) & (.029)\\
$Input\:Elasticity \left[\epsilon>med \right]_{h}$ &  & -.002* &  & .564*** &  & -.209*** \\
&  & (.001) &  & (.120) & & (.072)\\
$Input\:Rank_{hk}\: \times \: Input\:Elasticity\left[\epsilon>med \right]_{h}$ &  & .015  &  & 1.161* & & .072*\\
 &  & (.011) &   & (.095) & & (.038)\\
$Upstreamness_{hk}$ & -.001** & -.001** & .545***  & .521*** & -.228*** & -215***\\
 & (.001)  & (.001) & (.052) & (.065) & (.033) & (.040)\\
$Direct\:Requirement_{hk}$ & .001 & -.001 & 1.049*** & 1.036** & .027*** & .016** \\
 & (.001)  & (.001) & (.017)  & (.016) & (.008) & (.008)\\
$Constant$ & .006*** & .006*** & .003*** & .004*** & -2.738*** & -2.669*** \\
& (.001) & (.001) & (.001) & (.001) & (.038) & (.050)\\
 \hline
N. obs. & 7,805,667 & 1,257,911 &  7,805,667 & 1,257,911  & 7,805,667  & 1,257,911 \\
N. Parent companies  & 20,467  & 3,510 & 20,467  & 3,510 & 20,467 & 3,510\\
Parent fixed effects  & Yes & Yes & Yes & Yes & Yes & Yes\\
Clustered errors  & Yes & Yes & Yes & Yes & Yes & Yes\\
Adj. R-squared & .084  & .096  & - & - & - & - \\
Pseudo R-squared (McFadden's) & - & - & .150 & .181 & .160 & .191 \\
Log pseudo-likelihood & - & - & -40,445.334 & -30,347.22 & -39,945.341 & -29,985.011\\
\addlinespace
\hline
\bottomrule
\end{tabular}%
}
\begin{tablenotes}
\footnotesize
\singlespacing
\item Note: Odds ratios are reported when we run a logit model, in the form $\frac{Pr(y_{hr(k)}=1)}{Pr(y_{hr(k)}=0)}$. Marginal effects and beta coefficients are reported for probit specifications and linear probability models, respectively. Errors are clustered by I-O output industries. Variables are standardized at their mean and variance. *, **, *** stand for $p < .10$, $p < .05$, and $p < .01$, respectively. Please note, the variable $Input\:Elasticity\left[\epsilon>med\right]_{h}$ is available only for traded industries.
\end{tablenotes}
\end{table}

\end{document}